\documentclass[reqno]{amsart}
\usepackage{amsmath,amssymb}
\usepackage{graphics,epsfig,color}
\newtheorem{theorem}{Theorem}[section]
\newtheorem{proposition}[theorem]{Proposition}
\newtheorem{lemma}[theorem]{Lemma}

\theoremstyle{definition}
\newtheorem{definition}[theorem]{Definition}
\newtheorem{assumption}[theorem]{Assumption}
\newtheorem{remark}[theorem]{Remark}
\numberwithin{equation}{section}
\numberwithin{theorem}{section}
\renewcommand{\epsilon}{\varepsilon}
\newcommand{\ve}{\varepsilon}
\newcommand{\mc}[1]{{\mathcal #1}}
\newcommand{\bb}[1]{{\mathbb #1}}

\newcommand{\R}{\mathbb{R}}
\newcommand{\eps}{\varepsilon}
\newcommand{\id}{{1 \mskip -5mu {\rm I}}}
\newcommand{\varch}{\mathop{\rm ch}\nolimits}

\newcommand{\varash}{\mathop{\rm ash}\nolimits}
\newcommand{\de}{\mathop{}\!\mathrm{d}}

\title{A gradient flow approach to linear Boltzmann equations}
\author[G.\ Basile]{Giada Basile}
\address{Giada Basile \hfill\break \indent
   Dipartimento di Matematica, Universit\`a di Roma `La Sapienza'
   \hfill\break \indent
   P.le Aldo Moro 2, 00185 Roma, Italy}
 \email{basile@mat.uniroma1.it}
 \author[D.\ Benedetto]{Dario Benedetto}
 \address{Dario Benedetto \hfill\break \indent
   Dipartimento di Matematica, Universit\`a di Roma `La Sapienza'
   \hfill\break \indent
   P.le Aldo Moro 2, 00185 Roma, Italy}
 \email{benedetto@mat.uniroma1.it}
\author[L.\ Bertini]{Lorenzo Bertini}
\address{Lorenzo Bertini \hfill\break \indent
   Dipartimento di Matematica, Universit\`a di Roma `La Sapienza'
   \hfill\break \indent
   P.le Aldo Moro 2, 00185 Roma, Italy}
 \email{bertini@mat.uniroma1.it}

\begin{document}
\begin{abstract}
We introduce a gradient flow  formulation of linear Boltzmann equations. Under a diffusive scaling we derive a diffusion equation  by using the machinery of  gradient flows.
\end{abstract}


\maketitle
\thispagestyle{empty}

\section{Introduction}
The Boltzmann equation describes the evolution of the one-particle distribution on position and velocity of a rarefied gas. 
It has become a paradigmatic equation since it encodes most of the conceptual and technical issues in the description of the statistical properties for out of equilibrium systems.
In particular,  from a mathematical point of view, a global existence and uniqueness result 
is still lacking. In the kinetic regime, some transport phenomena can be described by linear Boltzmann equations. Typical examples are the charge (or mass) transport in the Lorentz gas \cite{Lo}, 
the evolution of a tagged particle in a Newtonian 
system in thermal equilibrium \cite{Sp}, and  the propagation of lattice vibrations in insulating crystals \cite{BOS}. 
Since the evolution equations are linear, their analysis is simpler. From one side, these equations have been derived from an underlying microscopic dynamics globally in time \cite{Ga, vBLLS, BGS-R, BOS}. 
From the other side, several results on the asymptotic behavior of the one-particle distribution have been obtained. In particular, by considering non degenerate scattering rates, under a diffusive rescaling 
the linear Boltzmann equation converges to the heat equation \cite{LK, BLP,  BSS, EP}.
 
In the present paper, inspired by the general theory in \cite{AGS}, we
propose a formulation of linear Boltzmann equations in terms of
gradient flows.  Recently there have been some attempts to formulate
the Fokker-Planck equation associated to continuous time reversible
Markov chains, equivalently homogeneous linear kinetic equations, as
gradient flows \cite{Ma,Mi,Er}, essentially in terms of energy
variational inequalities, and the potential applications of this
approach have yet to be fully investigated. The case of homogeneous
non linear Boltzmann equations is considered in \cite{Er2}.  The
present approach is based on an entropy dissipation inequality and can
be applied naturally to the inhomogeneous case, that appears novel. In
perspective, this approach could be adapted to the non linear, non
homogeneous Boltzmann equation.

We consider a  linear Boltzmann equation of the form
 \begin{equation}\label{BE}
  (\partial_t +b(v)\cdot \nabla_x )
  f(t,x,v)=
  \int\pi(\de v')\sigma(v,v')\big[f(t,x,v')- f(t,x,v)   \big]
\end{equation}
where $\pi(\de v)$ is a reference probability measure on the velocity space, $b$ is the drift, $\sigma(v,v')\geq 0$ is the scattering kernel and $f$ is the density of the one-particle distribution 
with respect to $\de x\,\pi(\de v)$. We assume the detailed balance condition, i.e. $\sigma(v,v')=\sigma(v',v)$. The entropy $\mathcal{H}(f)=\int \de x\int\pi(\de v)f\log f$ is a Lyapunov functional for the evolution 
\eqref{BE}, and the transport term do not affect its rate of decrease.  This observation will allow to formulate \eqref{BE} as the following entropy dissipation inequality
\begin{equation}\label{EDI}
 \mathcal{H}\big(f(T)\big) +\int_0^T \de t\, \mathcal{E}\big( f(t)\big) + \mathcal R_0 (f)\leq \mathcal H \big(f(0)\big),
\end{equation}
where $\mathcal E\big( f\big)=\displaystyle\int \de x\,\iint \pi(\de
v)\pi(\de v')\sigma(v,v')\big[\sqrt{f(x,v')}- \sqrt{f(x,v)} \big]^2$
is the Dirichlet form of the square root of $f$ and $\mathcal R_0\geq
0$ is a kinematic term that will be defined later.

As an application of \eqref{EDI} we  discuss the diffusive limit of the linear Boltzmann equation. This is a classical result, 
but the gradient flow formulation provides a transparent proof and allows to consider more general 
initial conditions, which are only required to satisfy an entropy bound.
More precisely, we will show that in the diffusive scaling limit the particle density
%$\rho^\ve(t,x)=\int\pi(\de v)f(\ve^{-2}t,\ve^{-1}x, v)$
converges 
to the solution of the heat equation. The proof will be achieved by taking the limit in the rescaled entropy dissipation inequality and deducing the corresponding inequality
for the heat equation.

\section{A gradient flow formulation}

In this section we introduce a gradient flow formulation of
non-homogeneous linear kinetic equations. Both for ease of
presentation and for future use, we however first review the gradient
flow formulation of the heat equation, that is here considered in
somewhat different setting that includes the current as a dynamical
variable.

Throughout the whole paper, the space domain is the $d$-dimensional
torus $\bb T^d:= \bb R^d/\bb Z^d$ and we denote by $\de x$ the Haar
measure on $\bb T^d$. The set of Borel probability measures on $\bb
T^d$ is denoted by $\mc P(\bb T^d)$ that we consider endowed with the
(metrizable) topology induced by the weak convergence. Recall finally
that the \emph{entropy} is the convex lower semicontinuous functional
$H\colon \mc P(\bb T^d) \to [0,+\infty]$ defined by $H(\mu) =
\int\!\de x\, \rho \log \rho$ if $\de\mu=\rho\,\de x$ and $H(\mu)=+\infty$
otherwise.

\subsection*{Heat equation}
We start by an informal discussion. 
Consider the heat equation on $\bb T^d$
\begin{equation*}
 \partial_t \rho=\nabla\cdot D\nabla \rho  
\end{equation*}
where $\rho$ is a probability density and the diffusion coefficient
$D$ is a positive symmetric $d\times d$ matrix. 
We introduce the currents as vector fields on $\bb T^d$, denoted by
$j$.  Given $\rho$, we define the associated current $j^\rho:=
-D\nabla\rho$. We can then rewrite the heat equation as
 \begin{equation}\label{he0}
 \begin{cases}
 \partial_t\rho+\nabla \cdot j=0\\
 j=j^\rho.
 \end{cases}
 \end{equation}
We shall rewrite this system as a variational inequality that
expresses the decrease of the entropy. 
 
Fix $T>0$. On the set of paths $(\rho(t),j(t))$, $t\in [0,T]$, 
satisfying the continuity equation $\partial_t \rho+\nabla \cdot j=0$
consider the action functional
\begin{equation*}
  I(\rho,j) = \frac 12 \int_0^T\!\de t \int\! \de x \frac 1{\rho(t)} 
  \big[ j(t) +D \nabla\rho(t) \big] \cdot D^{-1}
  \big[ j(t) +D \nabla \rho(t) \big],
\end{equation*}
where $\cdot$ denotes the inner product in $\bb R^d$.
This functional arises naturally by analyzing the large deviation
asymptotics of $N$ independent Brownians \cite{DG, KO} and its connection
with the gradient flow formulation of the heat equation is discussed
in \cite{ADPZ}. To be precise, the rate function in \cite{DG, KO, ADPZ} does not include the current as a
 dynamical variable but it can be extended
to this case, see \cite{BDGJL} for a similar functional  in the context of stochastic lattice gases.

Observe that $I\geq 0$ and $I(\rho, j)=0$ if and only if
$j=j^\rho$. Hence the second equation in \eqref{he0} is equivalent to
$I(\rho, j)\leq 0$.  By expanding the square we deduce 
\begin{equation}\label{ineq-calore}
\int_0^T\!\de t \int\!\de x\,
  \Big[ \frac 12 \frac 1{\rho(t)} j(t) \cdot D^{-1} j(t)  
  + \frac 12 \frac 1{\rho(t)}   \nabla\rho(t) \cdot D \nabla \rho(t)
  + \frac  1{\rho(t)} \nabla \rho(t) \cdot j(t) \Big]\leq 0.
\end{equation}
Since $(\nabla \rho)/\rho =\nabla \log\rho$, integrating by
parts and using the continuity equation, the last term is the total
derivative of $H(\rho(t))$.

We now introduce the \emph{Fisher information} $E$ as the Dirichlet
form of square root, namely
\begin{equation*}
  E(\rho)=\frac 12\int\! \de x  \frac 1{\rho}   \nabla\rho \cdot D \nabla
  \rho = 2\int \,\de x\, \nabla\sqrt\rho\cdot D\nabla\sqrt\rho. 
\end{equation*}
Let also the \emph{kinematic term} $R$ be the functional on the set  the
path $(\rho(t),j(t))$ defined by 
\begin{equation*}
 R (\rho, j)=\frac 12\int_0^T\! \de t\,\int\!\de x\,
    \frac 1{\rho(t)} j(t) \cdot D^{-1} j(t), 
\end{equation*}
then \eqref{ineq-calore} reads
\begin{equation}
  \label{gfhe}
 H(\rho(T))+\int_0^T \!\de t\,  E(\rho(t))+ R (\rho, j) \leq H(\rho(0)) 
\end{equation}
which is the gradient flow formulation of the heat equation that we
will use here. 

We now specify the precise formulation in which we consider a family of probabilities $\mu_t(\de x)=\rho(t,x)\de x$, 
$t\in [0, T]$, while the currents are the vector valued measures $J(\de t, \de x)=j(t,x)\de t \de x$.
Given $T>0$ let $C\big([0,T]; \mc P (\bb T^d)\big)$ be the set of continuous paths on
$\mc P(\bb T^d)$ endowed with the topology of uniform convergence. 
%and denote by $C_\mathrm{be} \big( [0,T] ; \mc P( \bb T^d)\big)$ the subset of $C\big([0,T]; \mc P (\bb T^d)\big)$
%such that $H(\mu_t)<+\infty$ for any $t\in [0,T]$. 
Let also  $\mc M\big([0,T]\times \bb T^d; \bb R^d \big)$  be the set of
vector valued Radon measures on $[0,T]\times \bb T^d$ endowed with
the weak* topology.  
Set $S:=C\big([0,T]; \mc P (\bb T^d)\big)\times \mc M\big([0,T]\times
\bb T^d; \bb R^d \big)$ endowed with the product topology.
%and denote by $S_\mathrm{be}:=C_\mathrm{be} \big( [0,T] ; \mc P( \bb T^d)\big)\times \mc M\big([0,T]\times \bb T^d; \bb R^d \big)$.

Given a positive $d\times d$ matrix $D$, the Fisher information 
$E\colon \mc P(\bb T^d)\to[0,\infty]$ can be defined by the variational formula
\begin{equation}
\label{varD-calore}
  E (\mu)=2\sup_{\phi\in C^2(\bb T^d)}
  \Big\{ -\int \!\de \mu \, e^{-\phi}\nabla \cdot D \nabla e^\phi\Big\},
\end{equation}
which implies its lower semicontinuity and convexity.
%%% regolaritÃ  di D?

The kinematic term $ R \colon S\to [0,\infty]$ admits the
variational representation 
\begin{equation}\label{varR-calore}
 R (\mu,J)=\sup_{w\in C([0,T]\times \bb T^d; \bb R^d)}
\Big\{J(w)-\frac 12 \int_0 ^T \!\de t\int \!\de\mu_t\,w\cdot Dw \Big\},
\end{equation}
which implies its lower semicontinuity and convexity.

\begin{definition}
\label{he}
Let $\nu\in\mc P(\bb T^d)$ with $H(\nu)<+\infty$. A path $(\mu, J)\in
S$ is a solution of the heat equation with initial condition $\nu$
iff $\mu_0=\nu$ and
\begin{eqnarray}
  \label{eqcont}
 &&\int_0^T \!\de t \,\mu_t(\partial_t\phi)   +J(\nabla \phi)=0,\qquad
 \phi\in C_c^1\big((0,T)\times \bb T^d  \big)\\
\label{ineq1-calore}
 && H(\mu_T)+\int_0^T \!\de t\, E(\mu_t)+ R (\mu, J) \leq H(\nu).
\end{eqnarray}
\end{definition}
%Observe that if \eqref{ineq1-calore} holds, then the same inequality
%also holds when $T$ is replaced by any $t\in [0,T]$. 

The standard formulation of the heat equation as gradient flow of the
entropy is recovered from \eqref{ineq1-calore} by projecting on
the density. Indeed, by the Benamou-Brenier lemma \cite{BeBr}, we deduce that if 
$(\mu,J)$ is a solution to the heat equation according to
Definition~\ref{he}, then $\mu=(\mu_t)_{t\in[0,T]}$ satisfies 
\begin{equation}
\label{hemu}
  H(\mu_T)+\int_0^T \!\de t\, \Big\{ E(\mu_t)+ 
  \frac 12 \big| \dot \mu_t  \big|^2  \Big\} \leq H(\nu)
\end{equation}
where $ \big| \dot \mu_t  \big|$ is the metric derivative of
 $t\mapsto \mu_t$ with respect to the Wasserstein-$2$ distance,
namely $\big| \dot \mu_t  \big| =\lim_{h\to 0}
\mathrm{d}_{W_2}\big(\mu_{t+h},\mu_t\big)/h$, 
where $\mathrm{d}_{W_2}$ denotes the Wasserstein-$2$ distance on $\mc
P(\bb T^d)$. 

Conversely, let $\mu$ be a solution to \eqref{hemu} satisfying
$\mu_0=\nu$. Introduce the functional $J^\mu$ on $C^1\big([0,T]\times
\bb T^d;\bb R^d\big)$ defined by 
$J^\mu(w) = \int_0^T\!\de t\,\mu_t \big(\nabla\cdot D w_t\big)$. 
Since $\int_0^T\!\de t \, E(\mu_t) \le H(\nu)$, the functional $J^\mu$
extends to an element of $\mc M\big( [0,T]\times \bb T^d;\bb
R^d\big)$, still denoted by $J^\mu$. Using again the Benamou-Brenier lemma
it is then straightforward to check that the pair $(\mu,J^\mu)$ is a
solution to the heat equation in the sense of Definition~\ref{he}.   

The previous remarks, together with the existence and uniqueness result
for the formulation \eqref{hemu} in \cite{Gi}, imply the following
statement. 

\begin{proposition}\label{t:uniq}
  For each $\nu\in\mc P(\bb T^d)$, with $H(\nu)< \infty$, there exists
  a unique solution of the heat equation with initial condition $\nu$.
\end{proposition}

\subsection*{Linear Boltzmann equations}

We do not need any particular hypotheses on the velocity space $\mc V$
that is assumed to be a Polish space, i.e.\ a metrizable complete and
separable topological space. We denote by $\mc P(\bb T^d\times \mc V)$
the set of probabilities on $\bb T^d\times \mc V$, that we consider
endowed with the topology of weak convergence.  We suppose given a
Borel probability measure $\pi$ on $\mc V$, a symmetric scattering
kernel $\sigma$, i.e.\ a Borel function $\sigma\colon \mc V\times \mc
V \to [0,+\infty)$ satisfying $\sigma(v,v')=\sigma(v',v)$, $v,v'\in
\mc V$, and a \emph{drift} $b\colon \mc V\to \bb R^d$. Given $P\in \mc
P(\bb T^d\times\mc V)$, we denote by $\mc H(P)$ the relative entropy
of $P$ with respect to the probability $\de x\,\pi(\de v)$ namely, $\mc H(P)
= \iint \!\de x \pi(\de v)\, f \log f $ if $\de P=f \,\de x\,\pi(\de
v)$ and $\mc H(P)=+\infty$ otherwise.

Also in this case we start by an informal discussion.
Fix $T>0$. Given a path $(P_t)_{t\in[0,T]}$ on 
$\mc P(\bb T^d\times \mc V)$ with $\de P_t = f(t,x,v) \, \de x
\,\pi(\de v)$, we use the 
shorthand notation $f=f(t,x,v)$, $f'=f(t,x,v')$ and set 
\begin{equation}
  \label{etaf}
  \eta^f =\eta^f(t,x,v,v') 
  := \sigma (f-f') = \sigma(v,v') \big[ f(t,x,v)-f(t,x,v')\big].
\end{equation}
We then rewrite the linear Boltzmann equation \eqref{BE} in the form
\begin{equation}\label{continuita}
  \begin{cases}
  \big(\partial_t +b(v)\cdot \nabla_x \big) f(t,x,v)
  +\int\! \pi(dv')\, \eta(t,x,v,v')=0\\
  \eta = \eta^f
  \end{cases}
\end{equation}
We understand that the first equation has to be
satisfied weakly and we shall refer to it as the \emph{balance}
equation. 
%Denote by $\mc J$ the space of functions $\eta = \eta(t,x,v,v')$
%antisymmetric in the exchange $v,v'$. 
We are going to rewrite the condition $\eta=\eta^f$ as an inequality that
expresses the decrease of the relative entropy $\mc H$. To this end, given
$\varkappa\geq 0$ let $\Phi_\varkappa\colon \bb R_+\times\bb R_+\times \bb
R\to [0,+\infty)$ be the convex function defined by 
\begin{equation*}
  \Phi_\varkappa (p,q;\xi) :=\sup_{\lambda\in\bb R}\Big\{\lambda \xi -
  \varkappa p \big(e^\lambda -1) - \varkappa q \big(e^{-\lambda}
  -1\big)\Big\}
\end{equation*}
observing that given $p,q\in\bb R_+$ the map $\xi\mapsto
\Phi_\varkappa(p,q;\xi)$ is positive (take $\lambda=0$), and equal to
zero iff $\xi=\varkappa (p-q)$.
Explicitly, as few computations shows, $\Phi_\varkappa$ reads
\begin{equation}
  \label{Phi}
  \begin{split}
    \Phi_\varkappa (p,q;\xi)  &=
        \xi \Big[ \varash \frac \xi{2 \varkappa\sqrt{pq}}
      -\varash \frac {\varkappa (p-q)}{2 \varkappa\sqrt{pq}} \Big]
    \\
    & - \Big[ \sqrt{ \xi^2 + 4 \varkappa^2 pq} - 
      \sqrt{ \big[\varkappa(p-q)\big]^2 + 4 \varkappa^2 pq}\Big]
  \end{split}
\end{equation}
where we recall that $\varash(z) = \log(z+\sqrt{1+z^2})$. We note that if $\varkappa=0$ then $\Phi_0(p,q;0)=0$ while $\Phi_0(p,q;\xi)=+\infty$ if $\xi\neq 0$.

Fix a path $(f(t),\eta(t))$, $t\in[0,T]$ satisfying
$\eta(t,x,v,v')=-\eta(t,x,v',v)$ and the balance equation in
\eqref{continuita}.  
The condition $\eta(t)= \eta^{f(t)}$, $t\in[0,T]$ is equivalent to
\begin{equation}
\label{jj1}
  \mc I (f,\eta) := \int_0^T\!\de t \int\! \de x
  \iint\! \pi(\de v)\,\pi(\de v')\,
  \Phi_\sigma (f,f';\eta) \le 0.
\end{equation}
This functional is connected with the large deviations asymptotic of
a Markov chain on $\mc V$ with transition rates $\sigma(v',v)\pi(\de
v')$, see \cite{BFG, MPR}.

We next write 
\begin{equation}
  \label{psi}
  \Phi_\varkappa (p,q;\xi) = \Phi_\varkappa(p,q;0) 
  +\xi \frac{\partial}{\partial \xi}
  \Phi_\varkappa(p,q;0)  +\Psi_\varkappa(p,q;\xi).
\end{equation}
By few explicit computations,
\begin{equation*}
  \begin{split}
    &\Phi_\varkappa(p,q;0) = \varkappa \big(\sqrt{p} -\sqrt{q}\big)^2\\
    &\frac{\partial}{\partial \xi} \Phi_\varkappa(p,q;0) 
    = \frac 12 \log \frac qp\\
    &\Psi_\varkappa(p,q;\xi) =
    \xi \varash \frac \xi{2 \varkappa\sqrt{pq}}
    - \Big[ \sqrt{ \xi^2 + 4 \varkappa^2 pq} - 2\varkappa \sqrt{pq} \Big].
  \end{split}
\end{equation*}
Observe that $\Psi_\varkappa$ has the variational representation
\begin{equation}
  \label{legpsi}
  \Psi_\varkappa(p,q;\xi) =\sup_{\lambda\in \bb R} \Big\{ \lambda \xi -
  2\varkappa \sqrt{pq} \big[ \varch \lambda -1\big] \Big\}.
\end{equation}
In particular, $\Psi_\varkappa\ge 0$. Moreover, while  the map
$(p,q;\xi)\mapsto \Psi_\varkappa(p,q;\xi)$ is convex, 
the map  $\xi\mapsto 
\Psi_\varkappa(p,q;\xi)$ is strictly convex. Finally, 
$\Psi_\varkappa(p,q;\xi)\sim \xi^2$ for $\xi$ small and 
$\Psi_\varkappa(p,q;\xi)\sim |\xi|\log|\xi|$ for $\xi$ large. 

Observe now that for the path $(f(t), \eta(t))$, $t\in[0,T]$
satisfying the balance equation in \eqref{continuita} we have 
\begin{equation*}
  \begin{split}
    \frac{d}{dt} \mc H(f(t)) &= \int\! \de x \int\! \pi(\de v) \, \log f 
    \Big[ - b(v)\cdot\nabla_x f
    - \int \pi(\de v') \, \eta(t,x,v,v')\Big]
    \\
    &= - \int\!\de x \iint \!
    \pi(\de v)\pi(\de v')\, \eta\log f
  \end{split}  
\end{equation*}
since the first term is a total derivative in $x$. 
Hence, by the antisymmetry of $\eta$,
\begin{equation*}
  \int\! \de x 
  \iint\! \pi(\de v)\pi(\de v') \,
   \eta \frac{\partial}{\partial \xi} \Phi_\varkappa(f,f';0)  
  = \frac{d}{dt} \mc H(f(t)).
\end{equation*}
for any $\varkappa >0$. 
Setting $\varkappa = \sigma$, 
inserting \eqref{psi} and integrating in time we obtain that,
for any $(f(t), \eta(t))$, $t\in[0,T]$ satisfying the balance
equation, it holds
\begin{equation}
  \label{uguh}
  \begin{split}
    \mc H(f(T)) +  \int_0^T\!\de t \int\! \de x \iint\!
    \pi(\de v) \pi(\de v') \,
    \big[ \Phi_{\sigma}(f,f';0)+\Psi_{\sigma}(f,f';\eta)\big]
    \\
    = \mc H(f(0)) + \int_0^T\!\de t \int \!\de x \iint\!
    \pi(\de v) \pi(\de v') \,
    \Phi_{\sigma}(f,f';\eta).
  \end{split}
\end{equation}
Gathering the above computations we conclude that \eqref{jj1} can be
rewritten as 
\begin{equation}
  \label{jj2}
  \mc H(f(T)) + \int_0^T\!\de t \, \mc E({f(t)}) + \mc R (f,\eta)
  \le \mc H(f(0))
\end{equation}
where
\begin{equation}
  \label{dfsr}
  \mc E({f}) = 
  \int \!\de x 
  \iint \!\pi(\de v)\pi(\de v')\, \sigma(v,v')
  \big[ \sqrt{f'} -\sqrt{f} \big]^2
\end{equation}
and
\begin{equation}
  \label{metric}
  \mc R (f,\eta) = \int_0^T\!\de t \int \!\de x 
  \iint\!\pi(dv)\pi(dv')\,
  \Psi_\sigma (f,f'; \eta).
\end{equation}
The inequality \eqref{jj2}, formally analogous to \eqref{gfhe}, is the
proposed gradient flow formulation of the linear Boltzmann equation \eqref{BE}.

\medskip
We now discuss the precise formulation in which we introduce the
measures $\de P = f(x,v) \, \de x\, \pi(\de v)$ and $\Theta(\de t,\de x,
\de v, \de v')= \eta(t,x,v,v') \, \de t\, \de x\, \de v\, \de v'$.  
We first specify the hypotheses on the scattering rate $\sigma$ and
the drift $b$ that are assumed to hold throughout the whole paper.

\begin{assumption}
\label{t:asb}$\phantom{i}$
 \begin{itemize} 
  \item[(i)] The \emph{scattering kernel} is a Borel function 
    $\sigma\colon \mc V \times \mc V\to [0,+\infty)$ satisfying
    $\sigma(v,v')=\sigma(v',v)$, $(v,v') \in \mc V\times \mc V$. 
  \item[(ii)] The \emph{scattering rate} $\lambda\colon \mc V \to
    [0,+\infty)$ is defined by $\lambda(v):=\int\!\pi(\de v')\,
    \sigma(v,v')$. We require that it has all exponential moments with
    respect to $\pi$ namely, $\pi\big[ e^{\gamma \lambda}\big]
    <+\infty$ for any $\gamma\in \bb R_+$.
  \item[(iii)] The \emph{drift} is a Borel function $b\colon \mc V \to
    \bb R^d$.  We require that it has all exponential moments with
    respect to $\pi$ namely, $\pi\big[ e^{\gamma |b| }\big]
    <+\infty$ for any $\gamma\in \bb R_+$, where $|b|$ is the
    Euclidean norm of $b$.
  \end{itemize}  
\end{assumption}

Given $T>0$ let $C\big([0,T]; \mc P (\bb T^d\times \mc V)\big)$ be the
set of continuous paths on $\mc P(\bb T^d \times \mc V)$ endowed with
the topology of uniform convergence. Denote by $\mc
M_\mathrm{a}\big([0,T]\times \bb T^d\times \mc V \times \mc V\big)$
the set of finite Radon measures on $[0,T]\times \bb T^d\times \mc V\times
\mc V$ antisymmetric with respect to the exchange of the last two
variables endowed with the weak* topology.  Set $\mc S:=C\big([0,T];
\mc P (\bb T^d \times \mc V)\big) \times \mc M_\mathrm{a}
\big([0,T]\times \bb T^d \times \mc V\times \mc V \big)$ endowed with
the product topology. 
%Denote by $\mc P_\mathrm{be}(\bb T^d\times \mc V)$ the set of probabilities with
%bounded entropy with respect to the probability $\de x\, \pi(dv)$. 
Let  
also $C_\mathrm{be}\big([0,T]; \mc P (\bb T^d\times \mc V)\big)$  the
set of paths $(P_t)_{t\in [0.T]}$ in $C\big([0,T]; \mc P (\bb
T^d\times \mc V)\big)$ such that $\sup_{t\in[0,T]}\mathcal H(P_t)<+\infty$ and let 
 finally 
$\mc S_\mathrm{be} :=C_\mathrm{be} \big([0,T];
\mc P (\bb T^d \times \mc V)\big) \times \mc M_\mathrm{a}
\big([0,T]\times \bb T^d \times \mc V\times \mc V \big)$

% and denote by $\mc S_0$ its closed subspace of
% paths satisfying weakly the balance equation, that is the collection
% of the elments $(P, \eta)\in \mc S$ such that
% \begin{equation}
%   \label{beq}
%   \int_0^T \!\de t \,P_t(\partial_t\phi + \nabla_x \phi \cdot b) 
%   = \frac 12 \int\! \eta(\de t,\de x,\de v, \de v') \, 
%   \big[\phi(t,x,v) -\phi(t,x,v')\big]
% \end{equation}
% for all continuos functions $\phi\colon (0,T)\times \bb T^d\times \mc
% V$ with compact support and continuously differentiable with respect
% to $t$ and $x$.  
If $P\in \mc P(\bb T^d\times \mc V)$ has finite entropy,
the Dirichlet form of the square root $\mc E$ can be defined by the
variational formula
\begin{equation}
\label{varD}
\mc E(P) := 2 \sup_{\phi \in C_\mathrm{b}(\bb T^d\times \mc V)} 
\iint \! P(\de x,\de v) \pi(\de v') \, \sigma(v,v') 
\Big[ 1 - e^{\phi(x,v')-\phi(x,v)} \Big].
\end{equation}
Note indeed the right hand side  is well defined for any $\phi\in C_\mathrm{b}(\bb T^d\times \mc V)$ in view of Assumption~\ref{t:asb} and the basic
entropy inequality $P(\psi ) \le \mc H (P) + \log \int\!\de x
\pi(\de v)\, e^\psi$, $\psi\colon \bb T^d\times \mc V\to \mathbb R$. 
The representation \eqref{varD} corresponds to the Donsker-Varadhan large deviation  for the empirical measure
of the continuous time Markov chain on $\mc V$ with transition rates $\sigma(v',v)\pi(\de v')$ \cite{DV}. Indeed,
$\mc E(P)=\sup_{\phi}\{-P(e^{-\phi}\mc L e^\phi) \}$, where
\begin{equation}\label{def:L}
 \mc L g (v)=\int \pi(\de v')\sigma(v',v)[g(v')-g(v)].
\end{equation}

%This variational representation readily implies  
%the convexity and lower semicontinuity of $\mc E$.

A variational representation for the kinematic term $\mc R$ is 
obtained by combining \eqref{legpsi} with the simple observation that
for $p,q\in \bb R_+$ we have $-2\sqrt{pq} =\sup_{a>0}\big\{ -
a p - a^{-1} q \big\}$. We thus let 
$\mc R \colon \mc S_\mathrm{be} \to [0,+\infty]$ be the functional defined by 
\begin{equation}
  \label{varR}
  \begin{split}
    \mc R (P,\Theta)
    &:= \sup_{\zeta,\alpha} \bigg\{
    \Theta (\zeta) - 
    \int_0^T\!\de t \iiint \! P_t(\de x,\de v)\pi(\de v') \, 
    \sigma(v,v') 
    \\ & \qquad \qquad
    \times \big[ \varch \zeta(t,x,v,v') -1 \big] \big[ \alpha(t,x,v,v') +
    \alpha(t,x,v',v)^{-1} \big]
    \bigg\},
\end{split}
\end{equation}
where the supremum is carried out over the continuous functions $\zeta
\colon [0,T]\times \bb T^d\times \mc V\times \mc V \to \bb R$ with
compact support and antisymmetric with respect to the exchange of the
last two variables and the bounded continuous functions $\alpha \colon
[0,T]\times \bb T^d\times \mc V\times \mc V \to (0,+\infty)$ uniformly
bounded away from zero. As before, the basic entropy inequality
implies that $\mc R$ is well defined.
%moreover it is convex and lower semicontinuous.   

At this point the gradient flow formulation of the linear Boltzmann
equations is simply specified by the following entropy dissipation
inequality. 

\begin{definition}
\label{t:dlbe}
Let $Q\in\mc P(\bb T^d\times \mc V)$ with $\mc H(Q)<+\infty$. An element
$(P, \Theta)\in \mc S_\mathrm{be}$ is a solution to the linear Boltzmann equation 
with initial condition $Q$ iff $P_0=Q$ and
\begin{eqnarray}
\label{beq}
&&\!\!\!\!\!\!\!\!\!\!\!\!\!\!\!\!\!\!\!\!
\displaystyle{  
\int_0^T \!\de t \,P_t(\partial_t\phi + b\cdot \nabla_x \phi) 
= \frac 12 \int\!\Theta(\de t,\de x,\de v, \de v') \, 
\big[\phi(t,x,v) -\phi(t,x,v')\big], }
\\
\label{ineq1-calore2}
&&\!\!\!\!\!\!\!\!\!\!\!\!\!\!\!\!\!\!\!\!
\displaystyle{  
  \mc H(P_T)+\int_0^T \!\de t\, \mc E(P_t)+ \mc R (P, \Theta) 
  \leq \mc H(Q).}
\end{eqnarray}
for all continuous functions $\phi\colon (0,T)\times \bb T^d\times \mc
V$ with compact support and continuously differentiable
in the first two variables.
\end{definition}

\begin{remark}
\label{t:rem}
  If $(P,\Theta)$ is a solution to the linear Boltzmann equation in
  the time interval $[0,T]$ then it solves the same problem   in the time
  interval $[0,t]$, $t\le T$ as well. This follows from the fact that any
  element $(P,\Theta)\in \mc S_\mathrm{be}$ satisfies, for $0\le s < t
  \le T$, the inequality
  \begin{equation}
    \label{jj1t}
      \mc H(P_t)+\int_s^t \!\de u\, \mc E(P_u)+ \mc R^{s,t} (P, \Theta_{[s,t]}) 
      \geq \mc H(P_s)
  \end{equation}
  where $\Theta_{[s,t]}$ is the restriction of $\Theta$ to the
  interval $[s,t]$ and the kinematic term $\mc R^{s,t}$ is defined as in
  \eqref{varR} with the interval $[0,T]$ replaced by $[s,t]$.  
  This inequality corresponds in fact to the trivial inequality $\mc
  I_{[s,t]}(P,\Theta)\ge 0$ where the action functional $\mc I_{[s,t]}$ is
  defined as in \eqref{jj1} with the interval $[0,T]$ replaced by
  $[s,t]$. The actual proof of \eqref{jj1t} is detailed  in Appendix \ref{app2}.
\end{remark}

It is of course possible to obtain a formulation only in terms of the
one particle distribution. More precisely, the formulation \eqref{EDI}
is obtained from \eqref{ineq1-calore2} simply by letting $\mc R_0
\colon C_\mathrm{be}([0,T]; \mc P(\bb T^d\times \mc V))\to [0,+\infty]$ be the
functional defined by $\mc R_0(P) = \inf_{\Theta} \mc R(P,\Theta)$ where
the infimum is carried out over all $\Theta\in \mc
M_\mathrm{a}([0,T]\times \bb T^d\times \mc V\times \mc V)$ such that
the pair $(P,\Theta)$ satisfies the balance equation \eqref{beq}. It is
however unclear to us whether $\mc R_0(P)$ could be represented by the
metric derivative of $t\mapsto P_t$ with respect to a suitable
distance on $\mc P(\bb T^d\times \mc V)$.

We now state an existence and uniqueness result for the above formulation, together
with a continuous dependence on the initial condition and the
coefficients. In particular uniqueness implies that solutions to  \eqref{BE}
with bounded entropy  are characterized by the gradient flow formulation  in Definition \ref{t:dlbe}.

\begin{theorem}
\label{t:eu}
For each $Q\in\mc P(\bb T^d\times \mc V)$, with $\mc H(Q)< +\infty$,
there exists a unique solution $(P,\Theta)$ to the linear Boltzmann
equation with initial condition $Q$. Furthermore
  \begin{itemize}
  \item[(i)] Set $\Theta^P(\de t,\de x, \de v, \de v'):=\de t\,\sigma
    (v, v')\big[P_t(\de x, \de v)\pi (\de v')-\pi (\de v) P_t(\de x,
    \de v') \big]$. Then $\Theta=\Theta^P$.
  \item [(ii)] 
    Let $\{Q^n\}\subset \mc P(\bb T^d\times \mc V)$ be such that
    $Q^n\to Q$ and $\mc H(Q^n)\to \mc H(Q)$ and denote by
    $(P^n,\Theta^n)\in \mc S_\mathrm{be}$ the solution to the linear Boltzmann
    equation with initial condition $Q^n$. Then the sequence
    $\{(P^n,\Theta^n)\}$  converges to $(P,\Theta)$.
  \item [(iii)] 
    Fix $Q\in \mc P(\bb T^d\times \mc V)$ with $\mc H(Q) <+\infty$ and
    consider coefficients $b$, $\sigma$ together with sequences $b^n$,
    $\sigma^n$ all satisfying Assumption~\ref{t:asb}.  Denote by
    $(P^n,\Theta^n)\in \mc S_\mathrm{be}$ the solution to the linear
    Boltzmann equation with initial condition $Q$ and coefficients
    $b^n$, $\sigma^n$.  If $b_n\to b$ in $\pi$ probability, 
    $\sigma^n\to \sigma$ in $\pi\times \pi$
    probability, and $\lim_n \big\{ \log \pi \big[ e^{\gamma |b^n -b |}
    \big] +\log \pi\big[ e^{\gamma |\lambda^n -\lambda|}\big] \big\}=
    0$ for any $\gamma>0$ then the sequence $\{(P^n,\Theta^n)\}$
    converges to $(P,\Theta)$.
  \end{itemize}
\end{theorem}

While uniqueness will be  proven by using  the argument in \cite{Gi},
the key ingredient for the continuity result is the following lemma. Its proof, whose details are omitted, is achieved by truncating with continuous and bounded functions and using the basic entropy inequality.

\begin{lemma}
  \label{t:lem}
  Let $\{P^n\}\subset \mc P(\bb T^d\times \mc V)$ be a sequence
  converging to $P$ and satisfying the entropy bound $\sup_n
  \mc H(P^n) < +\infty$. Then  $P^{n}(\phi) \to P(\phi)$ for any function
  $\phi$ having all exponential moments with respect to $\pi$.
  Moreover, if $\phi$ has all exponential moments and 
  $\lim_n \log \pi \big[ e^{\gamma |\phi^n -\phi|} \big] =0$ for any
  $\gamma>0$ then $P^{n}(\phi^n) \to P(\phi)$.
\end{lemma}

In view of the variational definition \eqref{varD}, this lemma readily
implies that the Dirichlet form $\mc E$ is lower semicontinuous on
sublevel sets of the entropy. Analogously, recalling \eqref{varR}, the
kinematic term $\mc R$ is lower semicontinuous on the sets $\big\{
(P,\Theta) \in \mc S_\mathrm{be} \colon \sup_{ t\in [0,T]} \mc H(P_t)
\le \ell \big\}$, $\ell \in \bb R_+$.

\begin{proof}[Proof of Theorem~\ref{t:eu}]
  We start by proving uniqueness, and in particular by showing that if $(P^1, \Theta^1)$ and $(P^2, \Theta^2)$ are solutions then $P^1=P^2$. Assume by contradiction that there exists $t\in (0,T]$ such that $P^1_t\neq P^2_t$ and let $(\bar P, \bar \Theta)=\frac 1 2 \big (P^1, \Theta^1 \big)+\frac 1 2 (P^2, \Theta^2)$. In view of Remark \ref{t:rem}, by the convexity of $\mathcal E$, $\mathcal R$ and the strict convexity of $\mc H$
    \begin{equation*}
      \mc H(\bar P_t)+\int_0^t \!\de s\, \mc E (\bar P_s) + \mc R ^{0,t}(\bar P, \bar \Theta)< \mc H (Q),  
    \end{equation*}  
    which by \eqref{jj1t} provides the desired contradiction.
    Uniqueness is now concluded  by observing that, for a given $P\in C_{\mathrm{be}}\big([0,T]; \mc P(\bb T^d\times\mc V)\big)$, the map $\Theta\mapsto \mc R(P, \Theta)$ is strictly convex, so that we can repeat the argument above with $P^1=P^2=P$ and deduce $\Theta^1=\Theta^2$.

    Postponing the proof of the existence, we show item {(i)}. We write $\Theta=\Theta^P+\tilde{\Theta}$ and we observe that, in view of the balance equation \eqref{beq},  $\tilde\Theta(\zeta)=0$ if  $\zeta(t,x,v,v')=z(t,x,v')-z(t,x,v)$ for some function  $z$. By choosing in the variational formula \eqref{varR} $\zeta=z'-z$, where $z=z(t,x,v)$ and $z'=z(t,x,v')$, we get
    \begin{equation*}
       \begin{split}
    \mc R (P,\Theta^P +\tilde\Theta)
    \geq &  \sup_{z,\alpha} \bigg\{
    \Theta^P (z'-z) - 
    \int_0^T\!\de t \iiint \! P_t(\de x,\de v)\pi(\de v') \, 
    \sigma(v,v') 
    \\ & \qquad \qquad
    \times \big[ \varch (z'-z) -1 \big] \big[ \alpha(t,x,v,v') +
    \alpha(t,x,v',v)^{-1} \big]
    \bigg\}\\
    = & \mc R (P,\Theta^P),
\end{split}
    \end{equation*}  
  where the last equality follows by direct computation. We conclude by uniqueness.

  We next prove item {(ii)}.  Remark~\ref{t:rem} implies
  that
  \begin{equation}
    \label{apb}
    \varlimsup_{n\to \infty} \, \sup_{t\in [0,T]} \, 
    \mc H(P_t^n) \le \mc H(Q).  
  \end{equation}
  In view of the lower semicontinuity of $\mc H$ and the observation
  after Lemma~\ref{t:lem} regarding $\mc E$ and $\mc R$, using the uniqueness
  it is enough to show precompactness of the sequence
  $\{(P^n,\Theta^n)\}\subset \mc S_\mathrm{be}$. 
  Observe indeed that, by Assumption~\ref{t:asb} and Lemma~\ref{t:lem},
  we can take the limit $n\to \infty$ in the balance equation
  \eqref{beq}.

  To prove precompactness of $\{\Theta^n\}$, observe that from
  \eqref{ineq1-calore2} and the variational representation
  \eqref{varR} it follows
  \begin{equation*}
    \sup_n\, \sup_{\zeta\colon \|\zeta\|_\infty\le 1} 
    \Theta^n(\zeta) < +\infty
  \end{equation*}
  and we conclude by the Banach-Alaoglu theorem. 

  The bound \eqref{apb} implies, by the coercive properties of the 
  relative entropy and Prohorov theorem, that there exist a compact $\mc K
  \subset\subset \mc P(\bb T^d\times \mc V)$ such that $P^n_t\in \mc
  K$ for any $n$ and $t\in[0,T]$. Hence,
  by Ascoli-Arzel\`a theorem, to prove the
  precompactness of $\{P^n\}$ it is enough to show that for each
  continuous $g\colon \bb T^d\times \mc V\to \bb R$ with compact support and continuously
  differentiable with respect to $x$ we have
  \begin{equation}
   \label{equic} 
    \lim_{\delta\downarrow 0} \, \sup_n\, \sup_{|t-s|<\delta} 
    \big| P^n_t(g) - P^n_s(g) \big| =0. 
  \end{equation}
  From the balance equation \eqref{beq} we deduce 
  \begin{equation*}
    \begin{split}
    & P^n_t(g) - P^n_s(g) = 
    - \int_s^t\!\de \tau\, P^n_\tau\big( b \cdot \nabla_x g \big) 
    \\
    &\qquad \qquad 
    -\frac 12 \int_{[s,t]\times \bb T^d\times \mc V\times \mc V} 
    \Theta^n(\de \tau, \de x, \de v, \de v')
    \big[ g(x,v) - g(x,v')\big].
    \end{split}
  \end{equation*}
  By Assumption~\ref{t:asb}, \eqref{apb}, and the basic entropy
  inequality, the first term on the right hand side vanishes as
  $|t-s|\to 0$ uniformly in $n$.
  On the other hand, by choosing $\alpha=1$ in the variational 
  representation \eqref{varR},
  \begin{equation*}
    \begin{split}
     & \Big| \int_{[s,t]\times \bb T^d\times \mc V\times \mc V} 
    \Theta^n(\de \tau, \de x, \de v, \de v')
    \big[ g(x,v) - g(x,v')\big] \Big| 
    \le  \mc R (P^n,\Theta^n)
    \\ 
    & + \quad 
    2 \int_s^t\!\de \tau \iiint \! P_\tau(\de x,\de v)\pi(\de v') \, 
    \sigma(v,v') 
    \big\{ \varch \big[ g(x,v) - g(x,v')  \big] -1 \big\}. 
  \end{split}
  \end{equation*}
  Replacing $g$ by $\gamma g$ with $\gamma>0$, using \eqref{apb},
  Assumption~\ref{t:asb}, the
  basic entropy inequality, and $\mc R (P^n,\Theta^n) \le \mc H(Q^n)$,
  we obtain that there exists a constant $C$ independent on $n$, $t,s$
  such that 
  \begin{equation*}
    \begin{split}
     & \Big| \int_{[s,t]\times \bb T^d\times \mc V\times \mc V} 
    \Theta^n(\de \tau, \de x, \de v, \de v')
    \big[ g(x,v) - g(x,v')\big] \Big| 
    \\ & \qquad 
    \le \frac 1\gamma \, \sup_n \mc H(Q^n)  
    + \frac C\gamma \, |t-s| \, \exp\{ 2 \gamma \| g\|_\infty\}. 
  \end{split}
  \end{equation*}
  By choosing $\gamma = (2 \|g\|_\infty)^{-1} \log
  (1/|t-s|)$ when $|t-s|\le 1$ the bound \eqref{equic} follows.

  In view of the second statement in Lemma~\ref{t:lem}, the proof of 
  item (iii) is achieved by the same arguments.
  
  We finally prove the existence result. Consider first the case in which
  $b$ and $\sigma$ are continuous and bounded. 
  If $Q(\de x, \de v) =f_0(x,v) \de x \pi(\de v)$ for
  some continuous density $f_0$ uniformly bounded away from zero, by
  classical results, the linear Boltzmann equation \eqref{BE} has a
  continuous solution $f(t,x,v)$ uniformly bounded away from zero. Set
  $P_t(\de x, \de v) := f(t,x,v) \de x \pi(\de v)$ and $\Theta(\de t,
  \de x, \de v, \de v') := \de t \de x \pi(\de v) \pi (\de v') \sigma
  (v,v') \big[ f(t,x,v)- f(t,x,v')\big]$.  It is then straightforward
  to justify the informal computations presented before and deduce
  that $(P, \Theta)$ solves the linear Boltzmann equation according to
  Definition~\ref{t:dlbe}.  Existence in the general case of $Q$
  satisfying the relative entropy bound $\mc H (Q) <+\infty$ and for
  coefficients $b$ and $\sigma$ satisfying Assumption~\ref{t:asb} is then
  achieved by items {(ii)} and {(iii)}.
\end{proof}

\section{Diffusive limit}

In this Section we discuss the asymptotic behavior of linear Boltzmann
equation, showing that in the  diffusive scaling limit  
the marginal distribution of the position evolves according to the
heat flow. This is a classical topic and has been much investigated in
the literature, see \cite{LK, BLP, BSS, EP} and \cite{DGP} for a more general setting.
We observe that this issue has
natural counterpart in probabilistic terms namely, the central limit
for additive functional of Markov chains. Indeed, the linear Boltzmann
equation \eqref{BE} is the Fokker-Planck equation for the Markov
process $(V_t,X_t)$ where $V_t$ is the continuous time Markov chain on
$\mc V$ with transition rates $\sigma(v',v)\pi(\de v')$ while $X_t$ is
the $\bb R^d$-valued  additive functional $X_t=\int_0^t\!\de s\,
b(V_s)$. We refer to \cite{KLO} for a recent monograph on this
topic. 

The gradient flow formulation of linear Boltzmann equations discussed
before allows a novel approach to the analysis of the diffusive limit.
According to a general scheme formalized in \cite{SaSe, Se}, a gradient flow formulation is particularly handy
for analyzing asymptotic evolutions and it does not require a direct analysis of the
dynamics.
Indeed, by comparing
Definition~\ref{t:dlbe} and Definition~\ref{he} we realize 
that the balance equation \eqref{beq} immediately leads to the
continuity equation \eqref{eqcont}. Moreover, taking into account
the convexity and lower semicontinuity of the entropy, in order to
establish the diffusive limit we only need to prove two limiting
variational inequalities comparing the Dirichlet form and the kinematic term
for the linear Boltzmann equation with the corresponding ones for the
heat flow. We shall prove these variational inequalities but, maybe
surprisingly, the two terms exchange their role in the diffusive
limit: the Dirichlet form $\mc E$ leads to  $R$ while
the kinematic term $\mc R$ leads to the Fisher information $E$.

To carry out the analysis of the diffusive limit of linear Boltzmann
equations a few extra conditions, implying in particular homogenization of the velocity, are needed.
As we show in the next section, this assumptions are satisfied for few natural models.
To this end, recalling that the scattering rate $\lambda$ is defined by $\lambda(v)=\int\pi(\de v')\sigma(v',v)$,
let $\tilde\pi$ be the probability on $\mc V$ defined by
\begin{equation}
  \label{tpi}
\tilde \pi(\de v) :=\frac{\lambda(v)}{\pi(\lambda)}\pi(\de v).
\end{equation}
\begin{assumption}$\phantom{i}$\label{assumpt1}
  \begin{itemize}  
 \item[(i)] The drift $b:\mc V\to\mathbb R^d$ is centered with respect to the
   measure $\pi$, namely 
   %, $b\in(L^1(d\pi);\bb R^d)$ and
   $\pi(b)=0$. 
   \item[(ii)] The scattering rate $\lambda$ satisfies $\pi[\lambda=0]=0$.
   \item[(iii)]$|b|^2/\lambda$ has all exponential moments, i.e. 
   $\pi[\exp \{\gamma |b|^2/\lambda\}]< +\infty$ for any $\gamma>0$.
   \item[(iv)] There exists a constant $C_0>0$
such that for any $g\in L^2(\tilde \pi)$
\begin{equation}
\label{pin}
  \int \de \tilde\pi \big[ g-\tilde\pi(g)\big]^2\leq C_0\iint\tilde\pi(\de v)\tilde\pi(\de v')\frac {\sigma(v,v')}
  {\lambda(v)\lambda(v')}\big[g(v)-g(v') \big]^2. 
\end{equation}
\end{itemize}
 \end{assumption}
 We remark that, as in the case of phonon Boltzmann equation, in  item (ii) we allow the case in which   $\lambda(v)=0$ for some $v\in \mc V$. 
Item (iv) corresponds to the assumption that  
the continuous time Markov chain with transition rates $\frac
{\sigma(v,v')} {\lambda(v)\lambda(v')}\tilde\pi(\de v')$ has spectral gap.
 The generator of this Markov chain  is $(K-\id)$, where
$K$ is given by
\begin{equation}\label{def:K}
\big(K g\big)(v)=\pi(\lambda)\int  \tilde\pi(\de v')\frac{\sigma(v,v')}{ \lambda(v)\lambda(v')}g(v').
\end{equation}  
Observe that
\begin{equation}
- \big(\mc L f\big)(v)=\lambda(v) \big[\big(\id-K\big) f\big](v),
\end{equation}
where $\mc L$ is the generator of the original Markov chain as defined in \eqref{def:L}. We emphasize that we do not assume the spectral gap of the generator $\mc L$, in fact 
the linear phonon Boltzmann equation, that will be discussed in the next section, meets the requirements in Assumption \ref{assumpt1} but its generator has no spectral gap.

Assumption \ref{assumpt1} implies that there exists $\xi\in L^2(\tilde \pi; \mathbb R^d)$ such that $-\mathcal L \xi =b$. 
Indeed, item (i) implies that $b/\lambda$ is centered with respect to $\tilde\pi$, item (ii) implies that $b/\lambda\in L^2(\tilde \pi; \mathbb R^d)$,
and finally item (iii) implies that $\xi:=(\id -K)^{-1}(b/\lambda)\in L^2(\tilde \pi; \mathbb R^d)$.

We need another technical condition on which we will rely to carry out a truncation on $\xi$.
\begin{assumption}\label{assumpt3}
One of the following alternatives holds
\begin{itemize}
 \item[(i)]$(-\mc L)^{-1}b$ is bounded, or
\item[(ii)] there exists $C<\infty$ such that $\big\|\big(\id -K\big)^{-1}f\big\|_{\infty}\leq C\|f\|_{\infty}$, for any $f$ such that $\tilde\pi(f)=0$.
\end{itemize}
\end{assumption}
Observe that if the map $(v,v')\mapsto\frac{\sigma(v,v')}{\lambda(v)\lambda(v')}$ is continuous and bounded then alternative (ii) holds.

For notation convenience, in this section we formulate the linear Boltzmann equation in terms of the density $(f,\eta)$,
where $\de P_t=f(t)\de x\pi(\de v)$ and $\de \Theta= \eta \de t\de x\pi(\de v)\pi (\de v')$. Indeed, the boundedness of the entropy implies the existence of $f$, while the boundedness of the metric term $\mc R$ implies the existence of $\eta$.

Let $\ve>0$ be a scaling parameter and consider the linear Boltzmann equation on a torus of linear size $\ve^{-1}$, equipped with the uniform probability distribution, on the time interval $[0, \ve^{-2}T]$. Under a diffusive rescaling of space and time, the rescaled solution 
  $(f^\ve, \eta^\ve)$ is defined on the torus of linear size one on the time interval $[0,T]$. It solves
\begin{eqnarray}
\label{constr}
&&\partial_t f^\ve(t,x,v)+\frac 1 \ve b(v)\cdot \nabla_x  f^\ve(t,x,v)+\frac 1 {\ve^2}\int\pi(\de v')
  \eta^\ve(t,x,v,v')=0\qquad\\
\label{ineq}
&&\mathcal H(f^\ve(T))+\frac 1 {\ve^2}\int_0^T \de t\, \mathcal E(f^\ve(t))+\frac 1{\ve^2} \mathcal R (f^\ve,\eta^\ve)\leq \mathcal H(f^\ve_0).\qquad
\end{eqnarray}

%%%%%%%%%%%%%%%%%%%%%%%%%%%%%%%%%%%%%%%%%%%%%%%%%%%%%%%%%%%%%%%%%%%%%%%%%%%%%%%%%%%%%%%%%%%%%%%%%%%%%%%%%%%%%%%%%%%%%%%%%

We set 
\begin{equation}
\label{drJ}
\begin{split}
 & \rho^\ve(t,x):=\int\pi (\de v) f^\ve(t,x,v)   \\
 & j^\ve(t,x):=\frac 1 \ve \int\pi(\de v)  f^\ve(t,x,v)b(v).
\end{split}
\end{equation}
Since $\eta^\epsilon(t,x,v,v')$ is antisymmetric with respect to the
exchange of $v$ and $v'$, by integrating \eqref{constr} with respect to
$\pi(dv)$ we deduce the continuity equation 
\begin{equation}
  \label{cont}
  \partial_t \rho^\epsilon + \nabla\cdot j^\epsilon = 0.
\end{equation}

\begin{theorem}
  \label{teo:3}
  Assume that $\rho^\ve_0\to\rho_0$ in $\mathcal P(\mathbb T^d)$ and
  $\lim_{\ve\to 0} \mathcal H(f^\ve_0)= H(\rho_0)$.  Then the sequence
  $(\rho^\ve, j^\ve)$ converges in $C([0,T]; \mathcal P(\mathbb
  T^d))\times \mathcal M ([0,T]\times \bb T^d; \mathbb R^d)$ to the
  solution to the heat equation as in Definition \ref{he}, with initial datum
  $\rho_0$ and diffusion coefficient
  \begin{equation}\label{D}
   D=\pi \big(b \otimes(-\mc L)^{-1}b\big).
  \end{equation}
\end{theorem}
%In order to prove the Theorem we will prove a compactness result (Lemma \ref{teo:3.0}) and pass to the limit in inequality \eqref{ineq}.
Note that, by  Assumption \ref{assumpt1}, $b/\lambda$ and $ \xi=(-\mc L)^{-1} b$ are in $ L^2(\tilde\pi; \mathbb R^d)$, hence the diffusion coefficient 
$D=\pi(\lambda)\,\tilde\pi\big((b/\lambda)\otimes\xi\big)$ is finite. 

The proof of this theorem will be achieved according to  the following strategy. We first show precompactness of the sequence $\{(\rho^\ve, j^\ve)\}$, 
we then consider a converging subsequence $(\rho^\ve, j^\ve)\to (\rho, j)$ and take the inferior limit in the inequality \eqref{ineq}. By the hypothesis of the theorem, 
$\mc H(f_0^\ve)\to H(\rho_0)$ and we  prove that the inferior limit of the left hand side of \eqref{ineq} majorizes the left hand side of \eqref{ineq1-calore}.
The statement follows by the uniqueness in Proposition \ref{t:uniq}. 
We introduce the following notations. 
If $(f^\ve,\eta^\ve)$ satisfy \eqref{constr}, \eqref{ineq},
recalling \eqref{tpi}, we set
\begin{equation}\label{def:u}\begin{split}
 & f^\ve(t,x,v)=u_\ve^2(t,x,v),\quad u_\ve(t,x,v)=\bar u_\ve(t,x) +
 \tilde u_\ve(t,x,v),\\
 & \bar u_\ve(t,x)=\tilde\pi (u_\ve(t,x,\cdot)).
\end{split}\end{equation}
We will use the following bounds that hold uniformly for $t\in [0,T]$.
By Cauchy-Schwarz inequality
\begin{equation}\label{ineq4}
  \int\!\de x\, \bar u_\epsilon^2(t) = \frac 1{\pi(\lambda)^2} 
  \int\!\de x
  \bigg(\int\!d\pi\, \lambda  \sqrt{f^\epsilon(t)}\bigg)^2 
  \le \frac {\pi(\lambda^2 )}{\pi(\lambda)^2}. 
\end{equation}
Moreover, by the basic
entropy inequality, for each $\gamma>0$
\begin{equation}\label{ineq1}
  \begin{split}
    \int \!\de x \!\int\! \de \pi  f^\ve(t) 
    \frac {|b|^2} {\lambda}
    &\le
    \frac 1\gamma \mc H (f^\epsilon(t)) + \frac 1\gamma \log \pi\Big(
    \exp\Big\{ \gamma  \frac{|b|^2}\lambda \Big\}\Big)
    \\
    &\le \frac 1\gamma \mc H (f^\epsilon_0) + \frac 1\gamma \log \pi\Big(
    \exp\Big\{ \gamma  \frac{|b|^2}\lambda \Big\}\Big).
   \end{split}
\end{equation}

\begin{lemma}
  \label{t:l4}
  There exists a constant $C$  such that for any   $\ve\in (0,1)$
  \begin{eqnarray}\label{ineq2}
 &&  \displaystyle\int \de x\int \de \pi
\,\tilde u_\ve(t)^2 \frac{|b|^2}{\lambda} < C,\qquad t\in[0,T],
  \\
  \label{ineq3}
&&\frac 1 {\ve^2}\int_0^T \de t\int dx \int \de \tilde\pi \tilde u_\ve(t)^2 <C .
  \end{eqnarray}
  \end{lemma}  
  
 \begin{proof}
In order to prove \eqref{ineq2}, 
since $\tilde u_\epsilon =\sqrt{f^\epsilon} -\bar
u_\epsilon$, for each $t\in [0,T]$
\begin{equation*}
%\label{1.4}
  \begin{split}
    %& 
    \int\!\de x\!\int\!\de \tilde\pi\, \tilde
    u_\epsilon^2(t) \frac {|b|^2}{\lambda^2} 
    %\\&\qquad 
    \le  \frac 2{\pi(\lambda)} 
    \int\de x \!\int\de  \pi\,  
    f^\ve(t) \frac{|b|^2}{\lambda}
    + \frac 2{\pi(\lambda)} 
    \int\de x\,  \bar u_\epsilon^2(t)
    \int\de \pi\,   \frac {|b|^2} {\lambda}.
 \end{split}
\end{equation*}
The first term on the right hand side 
is bounded  by \eqref{ineq1}, while
the second term is bounded, since 
$b^2/\lambda$ has finite exponential moments, by \eqref{ineq4}.

Regarding \eqref{ineq3},
by the Poincar\'e inequality \eqref{pin},
\begin{equation*}
  \begin{split}
    & \frac 1 {\epsilon^2} \int_0^T\!\de t\!\int\!\de x\!\int\!\de \tilde\pi\,
    \tilde u_\epsilon^2(t) 
    \\
    &\qquad \le \frac {C_0}{\epsilon^2 }\int_0^T
    \!\de t\!\int\de x \!\iint\tilde\pi(\de v)\tilde\pi(\de v') \frac {\sigma(v,v')}
    {\lambda(v)\lambda(v')} \big[u_\ve(t,x,v)- u_\ve (t,x,v')\big]^2
    \\
    &\qquad \le \frac{C_0}{\epsilon^2} \pi (\lambda)^2\int_0^T
    \de t\,\mathcal E(f^\ve(t))\leq C \mathcal H(f^\ve_0),
\end{split}
\end{equation*}
which concludes the proof.
\end{proof}

\begin{lemma}
  \label{teo:3.0}
  The set $\{(\rho^\ve, j^\ve)\}_{\epsilon\in (0,1]} \subset C([0,T]; \mathcal P(\mathbb T^d))\times \mathcal M
  ([0,T]\times \bb T^d; \mathbb R^d) $
  is
  precompact.
\end{lemma}

\begin{proof}
Given $0\le s <t \le T$,  the restriction of
the measure $\de J^\epsilon= j^\ve \de t \de x$ to $[s,t]\times \bb T^d$ is denoted by  $J^\epsilon_{s,t}$.
We will prove the
following bound. There exists a constant $C$ independent on $s,\,t$, such that
for any $\epsilon\in (0,1]$ and any $w\in C\big([s,t]\times \bb
T^d;\bb R^d\big)$
\begin{equation}
  \label{bonJ}
  \big|J^\epsilon_{s,t}(w)\big| = 
  \bigg|\int_s^t\!\de \tau\!\int\de x
  \, j^\epsilon(\tau,x)\cdot w(\tau,x)\bigg|
  \le C \,  \sqrt{t-s} \, \|w\|_\infty.
\end{equation}
Let us first show that it implies the statement. 
Choosing $s=0$, $t=T$, and applying the Banach-Alaoglu theorem,
\eqref{bonJ} directly yields the precompactness of $\{j^\epsilon\}$. 
Since $\mc P(\bb T^d)$ is compact, by the Ascoli-Arzel\`a theorem, 
to prove precompactness of $\{\rho^\epsilon\}$ it is enough to show
that for each $\phi\in C^1( \bb T^d)$
\begin{equation}\label{1.5}
  \lim_{\delta\downarrow 0} \, \sup_{\epsilon\in (0,1]} \, 
  \sup_{\substack{t,s\in [0,T]\\ |t-s|<\delta}}  \bigg| \int\!\de x\, 
  \big[\rho^\epsilon(t,x) -\rho^\epsilon(s,x)\big] \phi(x) \bigg| =0. 
\end{equation}
From the continuity equation \eqref{cont} we  deduce 
\begin{equation*}
  %\label{1.6}
  \int\!\de x\, 
  \big[\rho^\epsilon(t,x) -\rho^\epsilon(s,x)\big] \phi(x) 
  = \int_s^t\!\de \tau\!\int\!\de x\,  j^\epsilon(\tau,x)\cdot \nabla\phi(x)=
  J^\epsilon_{s,t}(\nabla\phi)
\end{equation*}
so that \eqref{1.5} follows readily from \eqref{bonJ}.

To prove \eqref{bonJ}, using the decomposition \eqref{def:u}, 
since $b$ has mean zero with respect to $\pi$,
by definition \eqref{drJ} of $j^\epsilon$ we get
\begin{equation*}
  \begin{split}
    &J^\epsilon_{s,t}(w) = 
    \frac 1\epsilon \int_s^t\!\de \tau\!\int\!\de x\!\int\!\pi(\de v)\, 
    f^\epsilon(\tau,x,v) \, b(v)\cdot w(\tau,x)\\
    &
    = \frac {\pi(\lambda)}\epsilon \int_s^t\!\de \tau\!\int\!\de x\!
    \int\!\tilde\pi(\de v)\, 
    \big[\tilde u_\epsilon^2(\tau,x,v) 
    + 2\bar u_\epsilon(\tau,x) \tilde u_\epsilon(\tau,x,v)\big]
    \, \frac{b(v)}{\lambda(v)}\cdot w(\tau,x).
  \end{split}
\end{equation*}
By Young's inequality, for each $\gamma>0$
\begin{equation*}
\begin{split}
& \frac 1{\ve} \tilde u_\ve^2 \frac {|b|}{\lambda} 
\le\frac \gamma{2\ve^2} \tilde u_\ve^2 + 
\frac 1{2\gamma} \tilde u_\ve^2\frac {|b|^2}{\lambda^2} \\
&\frac 2{\ve} \bar u_\ve \,|\tilde u_\ve| \frac {|b|}{\lambda}
\le \frac \gamma{\ve^2} \tilde u_\ve^2 + 
\frac 1{\gamma} \bar u_\ve^2\frac {|b|^2}{\lambda^2}
\end{split}
\end{equation*}
Then we obtain
\begin{equation*}
\label{split3}
  \begin{split}
|J^\epsilon_{s,t}(w)| \le 
\|w\|_\infty & \left\{  \frac {3\gamma}{2\epsilon^2}
    \int_0^T\!\de \tau\!\int\!\de x\!\int\!\de \tilde\pi\, \tilde
    u_\epsilon^2 + \frac 1{2\gamma}
    \int_s^t\!\de \tau\!\int\!\de x\!\int\!\de \tilde\pi\, \tilde
    u_\epsilon^2 \frac {|b|^2}{\lambda^2} + \right. \\
 &+ \left. \frac 1{\gamma}
    \int_s^t\!\de \tau\!\int\!\de x\, \bar
    u_\epsilon^2  \int\!\de \tilde\pi\, \frac {|b|^2}{\lambda^2} 
\right\}.
\end{split}
\end{equation*}
Using \eqref{ineq3}, \eqref{ineq2}, the fact that $|b|^2/\lambda$
has finite exponential moments, and \eqref{ineq4}, we obtain that there exists $C$ such that
\begin{equation*}
|J^\epsilon_{s,t}(w)| \le \frac C 2
\|w\|_\infty \left(  \gamma + \frac{t-s}{\gamma} \right),
\end{equation*}
then \eqref{bonJ} is obtained by choosing $\gamma = \sqrt{t-s}$.
\end{proof}

\begin{lemma}
  \label{teo:3.1}
  Assume that $\rho^\ve\to \rho$. Then for each $t\in [0,T]$
  \begin{equation*}
    \varliminf_{\ve\to 0} \mathcal H(f^\ve(t))\geq 
    H(\rho(t)). 
  \end{equation*}
\end{lemma}

\begin{proof}
  The statement is a direct consequence of the convexity and lower
  semicontinuity of the relative entropy.
\end{proof}

\begin{lemma}
  \label{teo:3.2}
  Assume that $(\rho^\ve, j^\ve)\to (\rho, j)$. Then
  \begin{equation}
    \varliminf_{\ve\to 0}\frac 1 {\ve^2}\int_0^T \de t \,\mathcal
    E(f^\ve(t))\geq R(\rho, j). 
  \end{equation}
\end{lemma}

\begin{proof}
Assume first that condition (ii) in Assumption \ref{assumpt3} holds.
  Recall \eqref{tpi} and observe that, in view of
  item (iii) in Assumption \ref{assumpt1}, $b/\lambda\in L^2(\tilde\pi)$. 
  Choose a
  sequence $\{a_n\}$, $a_n:\mc V\to\mathbb R^d$, converging to
  $b/\lambda$ in $L^2(\tilde\pi)$, such that: $a_n$ is bounded,
  $\tilde\pi(a_n)=0$ and $|a_n(v)|\leq |b(v)|/\lambda(v)$ for any $n\geq 1$.
  Upon extracting a subsequence, $a_n\to b/\lambda$ $\tilde\pi$-a.e.
  Set $\omega_{n}:=(\id-K)^{-1}a_n$. By the Poincar\'e inequality 
   \eqref{pin} $(\id -K)^{-1}$ is a  bounded
    operator on the subspace of $L^2 (\tilde\pi)$ orthogonal to the
   constants; hence $\omega_n$ converges to $\xi=(\id-K)^{-1}(b/\lambda)$ in
   $L^2(\tilde\pi)$. Moreover, by condition (ii) in Assumption \ref{assumpt3}, for each $n\geq 1$
   $\omega_n$ is bounded.

   Fix $w\in C([0,T]\times \bb T^d; \mathbb R^d)$.
   In the variational representation \eqref{varD} for $\mathcal E$  we
  chose the test function $\log\big(1+\ve w(t,x)\cdot \omega_{n}(v)\big)$, with $\ve$ small enough, and we deduce
  \begin{equation*}
    \begin{split}
    \frac 12 \frac 1 {\ve^2}\int_0^T \!\de t\, \mathcal E(f^\ve(t))\geq \frac 1
    {\ve}\int_0^T \! \de t \int\!\de x \int\!\de\pi 
    f^\ve\,\frac{w\cdot (-\mc L)\omega_n}{1+\ve w\cdot \omega_n}.
  \end{split}
\end{equation*}
Since $\omega_{n}$ is bounded, by Taylor expansion we obtain
\begin{equation*}\begin{split}
 \varliminf_{\ve\to 0}\frac 12 \frac 1 {\ve^2}\int_0^T \de t\, \mathcal E(f^\ve(t))\geq
 \varliminf_{\ve\to 0}\frac 1 \ve \int_0^T \de t \int\de x
 \int\de \pi  f^\ve \, w\cdot (-\mc L)\omega_{n}\\ 
 -  
 \varlimsup_{\ve\to 0}\int_0^T \de t \int\de x \int\de \pi
 f^\ve\, w\cdot \omega_{n} w\cdot (-\mc L)\omega_{n}. 
\end{split}\end{equation*}
Regarding  the first term on the right hand side, we write
\begin{equation*}
 \begin{split}
  \frac 1 \ve \int_0^T \de t \int\de x \int\de\pi  f^\ve\,
  w\cdot (-\mc L)\omega_{n}= 
   \int_0^T \de t \int\de x  j^\ve \cdot w + A_{\ve,n},
 \end{split}
\end{equation*}
 with 
\begin{equation*}
 A_{\ve,n}=\frac 1 \ve \int_0^T \de t \int\de x \int\de\pi  f^\ve\, w\cdot \big((-\mc L)\omega_{n} -b\big).
\end{equation*}
We will show that
\begin{equation}\label{lim_A}
 \lim_{n\to\infty} \sup_{\ve>0} |A_{\ve,n}|=0
\end{equation}
 and
 \begin{equation}
   \label{D2}
   \begin{split}
     \varlimsup_{n\to\infty}\varlimsup_{\ve\to 0}\int_0^T \de t
     \int\de x \int\de\pi f^\ve\, w\cdot \omega_{n}
     w\cdot (-\mc L)\omega_{n} \leq \int_0^T \de t \int\de x \rho w\cdot D\,w,
   \end{split}
 \end{equation}
 with $D$ given by \eqref{D}. Then, by optimizing over $w$ and using the variational representation
 \eqref{varR-calore}, the statement follows.
%%%%%%%%%%%%%%%%%%%%%%%%%%%%%%%%%%%%%%%%%%%%%%%%%%%%%%
%%%%%%%%%%%%%%%%%%%%% altro caso%%%%%%%%%%%%%%%%%%%%%

Postponing the proof of these two bounds, we
consider the case that  condition (i) in Assumption \ref{assumpt3} holds. Fix $w\in C\big([0,T]\times\mathbb T^d;\,\mathbb R^d\big)$, then in  the variational representation %\eqref{varD} 
for $\mathcal E$
we choose the test function $\log\big(1+\ve w(t,x)\cdot (-\mc L)^{-1}b(v)\big)$, with $\ve$ small enough. By Taylor expansion we deduce
 \begin{equation*}\begin{split}
 \varliminf_{\ve\to 0}\frac 12 \frac 1 {\ve^2}\int_0^T \de t \,\mathcal E(f^\ve(t))\geq
 \varliminf_{\ve\to 0}\frac 1 \ve \int_0^T \de t \int\de x
 \int\de\pi  f^\ve\, w\cdot b\\ 
 -  
 \varlimsup_{\ve\to 0}\int_0^T \de t \int\de x \int\de \pi
  f^\ve\, w\cdot (-\mc L)^{-1}b\, w\cdot b. 
\end{split}\end{equation*}
Recalling \eqref{drJ} and the variational representation \eqref{varD-calore}, it suffices to show
\begin{equation}
   \label{Db2}
   \begin{split}
     \varlimsup_{\ve\to 0}\int_0^T \de t
     \int\de x \int\de\pi f^\ve \,w\cdot b \,
     w\cdot (-\mc L)^{-1}b \leq \int_0^T \de t \int\de x \rho\, w\cdot D w.
   \end{split}
 \end{equation}

\smallskip
\noindent \emph{Proof of \eqref{lim_A}}.
According to the decomposition of $f^\ve $ in \eqref{def:u}, we write 
\begin{equation}
\begin{split}
    A_{\ve,n}=\frac 1\ve \int_0^T \de t \int\de x \int
    \de\pi \,\tilde u_\ve^2\, w\cdot \big((-\mc L)\omega_{n} -b\big)\\+ \frac 2
    \ve \int_0^T \de t \int\de x \int\de\pi\, \tilde u_\ve\,
    \bar u_\ve \,w\cdot \big((-\mc L)\omega_{n} -b\big),
\end{split}\end{equation}
where we used that $\pi[(-\mc L)\omega_n-b]=0$.  By Young's inequality,
for any $\gamma>0$ 
\begin{equation*}
\begin{split}
&\frac 1{\eps}  \tilde u_\ve^2 \,|w|\,\frac 1{\lambda} 
\big|(-\mc L)\omega_{n} -b\big|
\le \frac \gamma{2\eps^2}  \tilde u_\ve^2  + 
\frac 1{2\gamma} |w|^2\tilde u_\ve^2  \frac 1{\lambda^2} 
\big|(-\mc L)\omega_{n} -b\big|^2 ,\\
&\frac 2 \ve \frac 1{\lambda}\,\bar u_\ve |\tilde u_\ve|\,
   \,|w|\, \big|(-\mc L)\omega_{n} -b\big| 
\le \frac \gamma{\eps^2}  \tilde u_\ve^2+ \frac 1{\gamma}
 |w|^2\bar  u_\ve^2   \frac 1{\lambda^2} \big|(-\mc L)\omega_{n} -b\big|^2.
\end{split}
\end{equation*}
Then
\begin{equation}
\label{stima:A}
\begin{split}
 |A_{\ve,n}| \le &   \frac {3\gamma}{2\eps^2} 
 \int_0^T \!\de t \int \!\de x \int \!\de \tilde \pi\,  \tilde u_\ve^2  \\
  &+ 
\frac {\pi[\lambda]^2} {2\gamma} \|w\|_{\infty}^2
\int_0^T \!\de t \int\! \de x
\int \!\de \tilde \pi \, \tilde u_\ve^2 \frac 1 {\lambda^2}
\big| \big((-\mc L)\omega_{n} -b\big)\big|^2
\\
& + \frac {\pi[\lambda]^2} {\gamma}
 \|w\|_{\infty}^2\int_0^T \!\de t \int \!\de x \, \bar u_\ve^2
 \int \!\de \tilde \pi \,  \frac 1 {\lambda^2} \big| \big((-\mc L)\omega_{n} -b\big)\big|^2.
\end{split}\end{equation}
We claim that for each $\gamma >0$ the second and the third term
on the right hand side vanishes as $n\to \infty$ uniformly in
$\epsilon$. Since the first term on the right hand side can be bounded
by using \eqref{ineq3}, we then conclude taking the limit $\gamma \to 0$. 

To prove the claim, observe that,  by construction of the sequence $\{\omega_n\}$, 
\begin{equation}
\label{3.2}
  \begin{split}
  &[(-\mc L)\omega_n](v) = \lambda (v) a_n(v) \to b(v)\quad \pi\textrm{-a.e.}   \\
  & \big|(-\mc L)\omega_n(v) -b(v)\big| \le 2 |b(v)|.
  \end{split}
\end{equation}
As  $ \int\!\de x\, \bar u_\ve^2$ is bounded uniformly in $\ve$ by
\eqref{ineq4}, we conclude by dominated convergence and \eqref{ineq2}.

\smallskip
\noindent \emph{Proof of \eqref{D2}}.
It is enough to show that for each $n$
\begin{equation}
  \label{3.4}
  \lim_{\epsilon\to 0} 
  \int_0^T \!\de t \int\!\de x 
     \int \!\de \pi \, \big(f^\ve-\rho^\epsilon\big)\, w\cdot \omega_{n}
     w\cdot (-\mc L)\omega_{n}  =0.
\end{equation}
Indeed, by construction of the sequence $a_n$
\begin{equation*}\begin{split}
& \lim_{n\to\infty}\pi\big(\omega_n \otimes (-\mc L)\omega_n\big) =\lim_{n\to\infty}
\pi(\lambda)\tilde  \pi\big( \omega_n \otimes a_n\big)
=\pi(\lambda)\tilde\pi\big(\xi\otimes \tfrac b \lambda\big)=D.
\end{split}\end{equation*}

In order to prove \eqref{3.4}, by using  the decomposition of  $f^\ve$ in \eqref{def:u},
\begin{equation*}
  f^\epsilon -\rho^\epsilon = 
  2 \bar u_\epsilon \big[ \tilde u_\epsilon - \pi\big(\tilde u_\epsilon\big)
  \big] + \tilde u_\epsilon^2 - \pi \big(\tilde u_\epsilon^2\big).  
\end{equation*}
Since $\omega_n$ is bounded and  $|(-\mc L) \omega_n (v)| \le |b(v)|$, it
suffices
\begin{equation}
  \label{3.5}
  \begin{split}
  \lim_{\epsilon\to 0} 
  \int_0^T\!\de t\int\!\de x\, 
  \Big\{ 
  \bar u_\epsilon \pi\big( |\tilde u_\epsilon| {|b|}\big)
  + \bar u_\epsilon \pi\big(|\tilde u_\epsilon|\big)
  \pi\big(|b| \big)  + \pi\big( \tilde u_\epsilon^2 |b|\big)
  + \pi\big( \tilde u_\epsilon^2 \big)\pi\big(|b|\big)
  \Big\} =0.
  \end{split}
\end{equation}
By Cauchy-Schwarz, Lemma \ref{t:l4}, and \eqref{ineq4}, we directly
conclude that the first and third term vanishes as $\epsilon\to 0$.
To analyze the fourth term, given  $\delta>0$ we write
\begin{equation*}\label{g4}
  \begin{split}
 & \int_0^T \de t\,\int \!\de x\, \pi\big(\tilde u_\ve ^2\big)
 =\int_0^T \de t\,\int\!\de x\, \pi\big(\tilde u_\ve ^2 \,\chi_{\{\lambda\geq \delta\}}\big)
  +\int_0^T dt\,\int\!\de x\, \pi\big(\tilde u_\ve ^2 \,\chi_{\{\lambda<\delta\}}\big).
\end{split}\end{equation*}
By \eqref{ineq3}, the first term on the right hand side vanishes as $\ve\to 0$. It is therefore enough to show that the second term vanishes as $\delta\to 0$ uniformly in $\ve$. 
To this end, recalling that $\tilde u_\ve=u_\ve-\bar u_\ve$ with $u_\ve ^2=f^\ve$, 
\begin{equation*}%\label{g5}
  \begin{split}
    & %\int_0^T dt\,
    \int\!\de x\, \int\!\de \pi \, \tilde u_\ve ^2 \,\chi_{\{\lambda< \delta\}}
    =\int\!\de x\, \int\! \de \pi \, \big(u_\ve-\bar u_\ve \big)^2 \,\chi_{\{\lambda< \delta\}}\\
   & \leq 2  \int\!\de x\, \int\!\de \pi \, f^\ve \,\chi_{\{\lambda< \delta\}}
    + 2\,\pi\big(\lambda< \delta  \big)\int\!\de x\, \bar u _\ve ^2 ,
  \end{split}
\end{equation*}
and we conclude by using   \eqref{ineq4},  the basic entropy inequality and the assumption $\pi (\lambda=0)=0$.
To complete the proof of \eqref{3.5}, we  observe that by Schwartz inequality and \eqref {ineq4} the previous argument also implies that the second term vanishes as $\ve\to 0$.

\smallskip

\noindent
\emph{Proof of \eqref{Db2}}. As before, it suffices  to show that 
\begin{equation*}
  \lim_{\epsilon\to 0} 
  \int_0^T \de t\int\!\de x 
     \int\!\de \pi [f^\ve-\rho^\epsilon]\, w\cdot b\,
     w\cdot (-\mc L)^{-1}b  =0.
\end{equation*}
Since $(-\mc L)^{-1}b$ is bounded,  this follows from \eqref{3.5}.
\end{proof}

%%%%%%%%%%%%%%%%%%%%%%%%%%%%%%%%%%%%%%%%%%%%%%%%%%%%%%%%%%%%%%%%%%%%%%%%%%%%%%%%%%%%%%%%%%%%%
%%%%%%%%%%%%%%%%%%%% QUIIIIIIIIIIIIIIIIIIIIIIIIIIIIIIIIIIIIII
%%%%%%%%%%%%%%%%%%%%%%%%%%%%%%%%%%%%%%%%%%%%%%%%%%%%%%%%%%%%%%%%%%%%%%%%%%%%%%%%%%%%%%%%%%%%

\begin{lemma}
  \label{teo:3.3}
Assume that $(\rho^\ve, j^\ve)\to (\rho, j)$. Then
\begin{equation}
 \varliminf_{\ve\to 0}\frac 1{\ve^2}\mathcal R
 (f^\ve,\eta^\ve)\geq \int_0^T dt\, E(\rho(t)). 
\end{equation}
\end{lemma}

\begin{proof}
Assume first that condition (ii) in Assumption \ref{assumpt3} holds.
  Let $a_n$ and $\omega_n$ as in the previous lemma, and fix $\phi\colon (0,T)\times \bb T^d\to \bb R$ with compact support. In the variational formula \eqref{varR} we choose $\alpha=1$ and
$\zeta(t,x,v,v') = \epsilon \, \nabla \phi(t,x)\cdot \big(
  \omega_n(v')-\omega_n(v)\big)$,
  %antisymmetric with respect to the exchange of $v,v'$
then, by the antisymmetry of $\eta^\ve$ with respect to the exchange of $v,v'$,
\begin{equation*}
 \begin{split}
  \frac 1 {\ve^2} \mathcal R (f^\ve,\eta^\ve)\geq &
  -\frac 2{\epsilon}\int_0^T\! \de t\int\!\de x \, \nabla \phi(t,x) \cdot 
    \iint\!\pi(dv)\pi(dv') \eta^\epsilon(t,x,v,v')
    \omega_n(v)\\
    &-\frac 2 {\ve^2}   \int_0^T\! \de t\int\!\de x \, 
    \iint\!\pi(dv)\pi(dv')
    f^\ve(t,x,v)\sigma(v,v')\\
    & \qquad \times \big\{\varch \big(\epsilon \, \nabla \phi(t,x)\cdot \big(
  \omega_n(v')-\omega_n(v)\big)\big)-1  \big\}
 \end{split}
\end{equation*}
By the balance equation \eqref{constr},
\begin{equation*}
    \begin{split}
    &-\frac 2{\epsilon}\int_0^T\! \de t\int\!\de x \, \nabla \phi(t,x) \cdot 
    \iint\!\pi(dv)\pi(dv') \eta^\epsilon(t,x,v,v')
    \omega_n(v)\\
      &\quad = 
      -2 \epsilon  
      \int_0^T\!\de t\int\!\de x \, \partial_t\nabla \phi\cdot
      \int\!\de \pi f^\epsilon
      \,  \omega_n
      %\\
      %&\qquad\quad 
      - 2 \int_0^T\!\de t\int\!\de x\int\!\de \pi f^\epsilon
      \, \nabla\cdot \big[ \omega_n \cdot \nabla \phi \, b \big]. 
  \end{split}    
  \end{equation*}
Since $\omega_n$ is bounded, the first term on the right hand side
above vanishes as $\epsilon\to 0$. Therefore, by Taylor expansion of $\varch$,
\begin{equation*}
  \begin{split}
    &\varliminf_{\ve\to 0}\frac 1{\ve^2}
    \mathcal R(f^\ve,\eta^\ve)
    \\
    &\geq - \varlimsup_{\ve\to 0} 2 \int_0^T\!\de t\!\int\!\de x\!\int\!\pi(\de v)
    \,f^\epsilon(t,x,v) \nabla\cdot \big( \omega_n(v) \cdot \nabla
    \phi(t,x) \, b(v) \big)
    \\
    &%\qquad\;\;
    - \varlimsup_{\ve\to 0} \int_0^T\!\de t\!\int\!\de x\!\iint\!\pi(\de v)\pi(\de v')
    \,f^\epsilon(t,x,v) \sigma(v,v')%\\
     %&\qquad\qquad\qquad\qquad \times 
     \Big[ \nabla\phi(t,x) \cdot\big(\omega_n(v')-\omega_n(v)\big)\Big]^2.
   \end{split}
  \end{equation*}

We will show that 
\begin{equation}
  \label{lim1}
  \begin{split}
   \varlimsup_{n\to\infty} \varlimsup_{\ve\to 0} 2 \int_0^T\!\de t\int\!\de x\int\!\de \pi
  f^\epsilon \nabla\cdot \big[ \omega_n \cdot \nabla
  \phi \, b \big]
  \le 
  2\int_0^T\!dt\int\!dx\, \rho \nabla \cdot \big[D \nabla\phi\big]
\end{split}
\end{equation}
and
\begin{equation}
  \label{lim2}
\begin{split}
  &\varlimsup_{n\to\infty} \varlimsup_{\ve\to 0} 
  \int_0^T\!\de t\!\int\!\de x\!\int\!\pi(\de v)\pi(\de v')
  \,f^\epsilon(t,x,v) \sigma(v,v')
  %\\
  %&\qquad\qquad\qquad\qquad \times 
  \Big[ \nabla\phi(t,x)
  \cdot\big(\omega_n(v')-\omega_n(v)\big)\Big]^2
  \\
  &\qquad\qquad \le 2\int_0^T\!\de t\int\!\de x\, \rho(t,x) \nabla\phi(t,x)\cdot D
  \nabla\phi(t,x). 
\end{split}
\end{equation}
 Then, by optimizing over $\phi$ and using the variational representation
 \eqref{varD}, the statement follows.
 
 Assume now that condition (i) in Assumption \ref{assumpt3} holds. Then we choose as test functions $\alpha=1$ and  $\zeta(t,x,v,v')=\ve\nabla\phi(t,x)\cdot (-\mc L)^{-1}\big(b(v')- b(v)\big)$, 
 with a smooth $\phi\colon (0,T)\times\mathbb T^d\to\mathbb R^d$ with compact support. Using the fact that $(\mc L)^{-1}b$ is bounded we repeat the same arguments as above and therefore we have to show that  
\begin{equation}
  \label{lim1b}
  \begin{split}
   %&
  \varlimsup_{\ve\to 0} 2 \int_0^T\!\de t\int\!\de x\int\!\de \pi
  f^\epsilon \nabla\cdot \big[ (-\mc L)^{-1}b \cdot \nabla
  \phi \,b \big]
  %\\
  %&\qquad
  \le 
  2\int_0^T\!\de t\int\!\de x\, \rho \nabla \cdot \big[D \nabla\phi\big]
\end{split}
\end{equation}
and
\begin{equation}
  \label{lim2b}
\begin{split}
  &\varlimsup_{\ve\to 0} 
  \int_0^T\!\de t\!\int\!\de x\!\iint\!\pi(\de v)\pi(\de v')
  \,f^\epsilon(t,x,v) \sigma(v,v')
  %\\
  %&\qquad\qquad\qquad\qquad \times 
  \Big[ \nabla\phi(t,x)
  \cdot(-\mc L)^{-1}\big(b(v')-b(v)\big)\Big]^2
  \\
  &\qquad\qquad \le 2\int_0^T\!\de t\int\!\de x\, \rho(t,x) \nabla\phi(t,x)\cdot D
  \nabla\phi(t,x). 
\end{split}
\end{equation}
\smallskip
\noindent \emph{Proof of \eqref{lim1}}.
We claim that for each $n$
\begin{equation*}
  \lim_{\epsilon\to 0} 
  \int_0^T\!\de t\int\!\de x\int\!\de \pi
  \big[ f^\epsilon -\rho^\epsilon\big] 
  \nabla_\cdot \big[ \omega_n \cdot \nabla
  \phi b \big] =0, 
\end{equation*}
which  is proven exactly as \eqref{3.4} observing that there we used the bound
$|(-\mc L) \omega_n|\le |b|$.
Since $\rho^\epsilon\to \rho$ and, by construction of the sequence $\{\omega_n\}$,  $\lim_{n} \pi( \omega_n \otimes b)=
D$, we then conclude.

\smallskip
\noindent \emph{Proof of \eqref{lim2}}.
We first show that for each $n$
\begin{equation*}
  \begin{split}
  &\lim_{\ve\to 0} 
  \int_0^T\!\de t\!\int\!\de x\!\int\!\pi(\de v)\pi(\de v')
  \,\big[ f^\epsilon(t,x,v) -\rho^\epsilon(t,x)\big] \sigma(v,v')
  \\
  &\qquad\qquad\qquad\qquad \times 
  \Big[\nabla\phi(t,x)
  \cdot\big(\omega_n(v')-\omega_n(v)\big)\Big]^2 = 0.
  \end{split}
\end{equation*}
Since $\omega_n$ is bounded and $\lambda(v)=\int\!\pi(dv')
\sigma(v,v')$, it is enough to prove that 
\begin{equation*}
  \lim_{\ve\to 0} 
  \int_0^T\!\de t\!\int\!\de x\!\int\!\de\pi \lambda 
  \,\big| f^\epsilon -\rho^\epsilon\big| = 0,
\end{equation*}
whose proof is achieved by the same arguments used in the proof of \eqref{D2}.
We then conclude by observing that
\begin{equation*}\begin{split}
&\lim_{n\to\infty} \iint\!\pi(\de v)\pi(\de v')\sigma(v,v')\big(\omega_n(v')-\omega_n(v) \big)\otimes \big(\omega_n(v')-\omega_n(v) \big)\\
&\qquad =\lim_{n\to\infty}2\pi\big(\omega_n\otimes (-\mathcal L)\omega_n  \big)=2D.
\end{split}\end{equation*}

\smallskip
\noindent 
\emph{Proofs of \eqref{lim1b} and \eqref{lim2b}}. These are achieved as the proofs of \eqref{lim1} and \eqref{lim2}  with $(-\mc L)^{-1}b$ instead of $\omega_n$. 

\end{proof}

\section{Specific examples}
We consider three examples of linear Boltzmann equations and we show that they meet the requirements in the Assumptions \ref{t:asb}, \ref{assumpt1} and \ref{assumpt3}. 
For all of them the convergence to a diffusion is a classical result, 
therefore they are  suitable testers for the machinery. 
We emphasizes however that we do not  require  the initial distribution to be  in $L^2$, as it is usual in the classical approaches, 
but only with finite entropy with respect to the reference measure.

\subsection{Boltzmann-Grad limit for the Lorentz gas with hard scatterers}
The first example is the linear Boltzmann equation derived for the one particle distribution in \cite{Ga}, starting from the  Lorentz gas moving in a random array of 
fixed scatterers (hard spheres), in the  Boltzmann-Grad limit. Since collisions are elastic, the kinetic energy is preserved, therefore the phase space is $\mathbb T^d\times S_{|v|}^{d-1}$, 
$d\geq 2$,
where $S_{|v|}^{d-1}$ is the 
$d$-dimensional sphere with radius $|v|$. Without loss of generality, we assume $|v|=1$.
The equation then reads
\begin{equation}\label{ex1}
\partial_t f(t,x,v)+v\cdot \nabla_x f(t,x,v)= \int_{S^{d-1}} \de \hat n\,[\hat n \cdot v]_+ \big[ f(t,x,v')-f(t,x,v)\big],
\end{equation}
where $v'=v-2(v\cdot \hat n)\hat n$. The invariant measure  is the uniform measure $\de \hat n$ on $S^{d-1}$, the scattering rate $\lambda$ is equal to $c$, for some constant  $c$ 
depending on the dimension $d$.
In order to identify the scattering kernel, we consider here the case $d=2$, referring to Appendix \ref{s:A} for analogous computations if $d\geq 3$.   
By identifying the velocity $v\in S^1$ with  the angle $\theta$, 
we rewrite the previous equation as
\begin{equation*}
 \partial_t f(t,x,\theta)+b(\theta)\cdot \nabla_x f(t,x,\theta)=\frac 1 2  \int_{S^1} \de \theta'\,\Big|\sin\frac {\theta-\theta'}2 \Big|\big[ f(t, x,\theta')-f(t,x,\theta)\big],
\end{equation*}
with $b(\theta)=(\cos \theta,\, \sin\theta)$. In particular, the scattering kernel is  $\sigma(\theta,\theta')=\Big|\sin\frac {\theta-\theta'}2 \Big|$.  
Recalling the definition \eqref{tpi} of $\tilde\pi$ , the operator $K$ with kernel $\sigma/\lambda$ is compact in $L^2(\tilde\pi )$. Since $1$ is a simple 
eigenvalue of $K$, then  the modified chain has spectral gap. Moreover, as shown in \cite{BNPP}, Lemma 4.1, $L^{-1}v$ is bounded.
Hence Assumptions \ref{t:asb}, \ref{assumpt1} and the alternative (i) of Assumption \ref{assumpt3} hold.

\subsection{Rayleigh-Boltzmann equation}
The Rayleigh-Boltzmann equation, also known as linear Boltzmann
equation or Lorentz-Boltzmann equation, has been derived in the
Boltzmann-Grad limit by looking at the distribution of a tracer
particle in a gas of particles (hard spheres) in thermal equilibrium
\cite{vBLLS}. The velocity space is then $\mathbb R^d$, $d\geq 2$, and
the reference measure is the Maxwellian distribution with temperature
$\beta^{-1}$, whose density with respect to the Lebesgue measure is
denoted by $h_\beta$.  The equation reads
\begin{equation}\label{ex3}\begin{split}
  &\partial_t f(t,x,v)+v\cdot \nabla_x f(t,x,v)\\
  &\qquad= \int_{\mathbb R^d}\!\de v_1\,
  h_\beta(v_1)\int_{S^{d-1}}\!\de\hat n \,\big[\hat n\cdot(v-v_1)\big]_+\big[f(t,x,v')-f(t,x,v)  \big],
\end{split}\end{equation}
where $v'= v-\hat n\cdot(v-v_1)\,\hat n$. As shown in the Appendix \ref{s:A},
he scattering rate is
$$
\lambda(v)=\chi\int_{\mathbb R^d}\de v_1 h_\beta(v_1)|v-v_1|,
$$
where $\chi$ is the constant given by $\chi=\int_{S^{d-1}}\!\de \hat n [\hat n\cdot\hat v]_+$, in which $\hat v\in S^{d-1}$.  
In particular $\lambda$ is bounded away from $0$ and it has linear growth for 
large  $|v|$.
Therefore  $|v|$ and $|v|^2/\lambda(v)$ have all the exponential moments with respect to $h_{\beta}(v)\de v$.
In Appendix \ref{s:A} we identify the scattering kernel $\sigma$, see  \eqref{sfuturo}. From this expression and the properties of $\lambda$,
recalling the definition \eqref{tpi},
it follows that $K$ in \eqref{def:K}
has a kernel in  $L^2(\tilde\pi\times\tilde \pi)$.
Hence $K$ is compact in $L^2(\tilde\pi)$. Since $1$ is simple eigenvalue of $K$, then $(\id-K)$ has spectral gap. 
The previous statements imply that Assumptions \ref{t:asb} and \ref{assumpt1} hold.
The proof of alternative (i) in Assumption \ref{assumpt3} is the content of Appendix \ref{s:A}.

\subsection{Linear phonon Boltzmann equation}
The equation has been derived in the kinetic limit  starting from an harmonic chain of oscillator perturbed by a stochastic conservative noise \cite{BOS}. 
It describes  the evolution
of the energy density of the normal modes, or phonons, identified by a wave-number $k\in\mathbb T^d$. The velocity space is then $\mathbb T^d$. 
Let $\omega$ be the dispersion relation of the harmonic lattice, 
i.e. $\omega(k)=\big(\nu +4\sum_{i=1}^d\sin^2(\pi k_i) \big)^{1/2}$, where $\nu\geq 0$ is the intensity of the pinning.
A phonon with wave-number $k$ travels with velocity $\frac 1 {2\pi}\nabla\omega$, then it is scattered.
The corresponding Fokker-Planck equation reads
\begin{equation}\label{ex2}
 \partial_t f(t,x,k)+\frac 1 {2\pi}\nabla\omega(k)\cdot \nabla_x f(t,x,k)=\int_{\mathbb T^d}\!\de k'\,\sigma(k,k')\big[f(t,x,k') -f(t,x,k) \big],
\end{equation}
where the scattering kernel has the form $\sigma(k, k')=
\sum_{i=1}^d \sin^2(\pi k_i)\sin^2(\pi k'_i)$ for $d\geq 2$.
In dimension one it has a slightly different form, but despite the details the main features are that $\sigma$ is positive, bounded and symmetric in the exchange $k, k'$.
Then the invariant measure $\pi$ is the Haar measure on the torus.
Moreover $\sigma$ vanishes in zero, since  $\sigma(k,k')\sim |k|^2$ for small $k$, and the scattering rate $\lambda$ has the same behavior. 
More precisely, $\lambda=c \sum_{i=1}^d \sin^2(\pi k_i)$, for some constant $c>0$.
Since $\partial_i \omega(k)= 2\pi\sin(2\pi k_i)/\omega(k)$, $i=1,\dots,d$,
in order to guarantee that $|\nabla\omega|^2/\lambda$ has exponential moments we have to restrict to the case $\nu>0$.
This corresponds to assume that in the underlying harmonic chain the translational symmetry is broken,
due to the presence of a on-site potential (pinning).
Recalling the definition \eqref{tpi} of $\tilde\pi$, 
by the properties of $\sigma$ and $\lambda$ we deduce that the operator $K$ defined in \eqref{def:K} has a kernel of the form $p(k,k')\tilde\pi(dk')$, with $p$ strictly positive and bounded.
Then $K$ is a compact operator in $L^2(\tilde\pi)$, and since  $1$ is a simple eigenvalue, than the modified chain has spectral gap. The previous statements imply that Assumptions 
\ref{t:asb} and \ref{assumpt1} hold.
Finally, since the modified chain  satisfies the Doeblin condition, then for each $f\in L^\infty(\mathbb T^d)$ such that $\tilde\pi(f)=0$ we have 
$\|\sum_{n\geq 0} K^n f\|_{\infty}\leq c\|f\|_{\infty}$, which implies 
alternative (ii) in Assumption \ref{assumpt3}.

In the unpinned case $\nu=0$ the diffusion coefficient $D$ diverges in dimension $d=1,2$. In these cases the asymptotics of the linear phonon Boltzmann equation is in fact a superdiffusion when $d=1$ \cite{JKO, BaBo} and a diffusion under an anomalous scaling,
i.e. with  logarithmic  corrections,  when $d=2$ \cite{Ba}; see also \cite{MMM} for other models with super-diffusive behavior.
On the other hand, for $d\geq 3$ the diffusion coefficient $D$ is finite even if $|b|^2/\lambda$ does not have exponential moments.
Moreover, as discussed in \cite{Ba}, if the initial distribution satisfies suitable integrability conditions, the diffusive scaling holds. As the gradient flow approach here introduced requires only an entropy bound on the initial condition, it does not cover this case. It is not clear if this is just a limitation of the present approach or the diffusive limit fails if the integrability conditions are not satisfied. Indeed, phonons with small wave number are responsible of the ballistic transport which, in dimension one and two, induces the superdiffusion.
If the initial conditions gives enough weight to those phonons, similar effects might occur also for $d\geq 3$.

\appendix
\section{Entropy balance}\label{app2}
We here prove \eqref{jj1t}.
We can assume that its left hand side  is finite. Using
also that $P_r$ has bounded entropy for each $r\in[s,t]$ we deduce
that $P_r(\de x, \de v) = f_r(x,v) \de x \, \pi(\de v)$ and 
$\Theta_{[s,t]} (\de r, \de x, \de v, \de v') = \eta_r(x,v,v') \de r
\, \de x\, \pi(\de v) \pi(\de v')$. Moreover, recalling the function
$\Psi_{\varkappa}$  in \eqref{legpsi},
\begin{equation*}
  \begin{split}
    &  \int_s^t \!\de r\, \mc E(P_r) = 
    \int_s^t \!\de r \int\!\de x \iint \!\pi(\de v)\pi(\de v')\,
    \sigma(v,v') \Big[ \sqrt{f_r(x,v')} -\sqrt{f_r(x,v)}\Big]^2 \\
    & \mc R^{s,t} (P, \Theta_{[s,t]})
     = 
    \int_s^t \!\de r \int\!\de x \iint \!\pi(\de v)\pi(\de v')\,
    \Psi_{\sigma(v,v')} \big( f_r(x,v),f_r(x,v');  \eta_r(x,v,v')
    \big) 
  \end{split}
\end{equation*}
and
\begin{equation}\label{intrep}
\int_s^t \!\de r\, \mc E(P_r) +  \mc R^{s,t} (P, \Theta_{[s,t]}) < +\infty.
\end{equation}  

We claim that, for $f$ and $\eta$ as above, the following entropy
balance holds
\begin{equation}\begin{split}
  \label{entbal}
 & \mc H(P_t) -\mc H(P_{s'})\\ & = \frac 12 
  \int_{s'}^t \!\de r\! \int\!\de x\! \iint \!\pi(\de v)\pi(\de v')\, 
  \eta_r(x,v,v') \big[ \log f_r(x,v')- \log f_r(x,v) \big]
\end{split}\end{equation}
for any $0\leq s< s'<t\leq  T$.
Note  that the last term is well defined by Legendre duality and
\eqref{intrep}.  Informally, it is deduced by
choosing the test function $\id_{ [s',t]}(r)\log f_r(x,v)$ in the balance equation
\eqref{beq}. The actual proof is carried out by a truncation argument
that is next detailed.

\smallskip
\noindent\emph{Step 1. Approximation by space time convolutions.}
%Extend  $ [s,t]\ni r\mapsto f_r$ and  $[s,t]\ni r\mapsto \eta_r$ to functions defined on
%% the whole real line by setting $f_r = f_s$ for $r<s$, $f_r= f_t$ for
%% $r>t$ and $\eta_r =0$ for $r\not\in [s,t]$. Observe that the balance
%equation \eqref{beq} still holds for these extended functions.
For
$n\in\bb N$ let now $\chi_n \colon \bb R \to  \bb R_+$ be a smooth approximation of the identity with compact support contained in the positive axis, and 
$g_n\colon \bb T^d\to \bb R_+$ be a smooth approximation of the identity.
For $0\leq s< s'\leq r \leq t \leq T$, by choosing $n$ such that the $\mathrm{supp}\chi_n \subset [0,s'-s]$,
 we define    
\begin{equation*}
  \begin{split}
    f^n_r(x,v) := & \int\!\de r' \int \de y \, \chi_n(r-r') g_n(x-y)
    f_{r'}(y,v) \\
    \eta^n_r(x,v,v') := &\int\!\de r' \int \de y \, \chi_n(r-r') g_n(x-y)
    \eta_{r'}(y,v,v'). 
  \end{split}
\end{equation*}
As simple to check, the pair $(f^n,\eta^n)$ satisfies the balance
equation and, by \eqref{intrep} and convexity, there exists a constant
$C$ such that 
\begin{equation*}
  \begin{split}
    &\int_{s'}^t \!\de r \int\!\de x \iint \!\pi(\de v)\pi(\de v')\,
    \sigma(v,v') \Big[ \sqrt{f^n_r(x,v')} -\sqrt{f^n_r(x,v)}\Big]^2 
    \\
    &\qquad +
     \int_{s'}^t \!\de r \int\!\de x \iint \!\pi(\de v)\pi(\de v')\,
    \Psi_{\sigma(v,v')} \big( f^n_r(x,v),f^n_r(x,v');  \eta^n_r(x,v,v')
    \big)   \le C,
  \end{split}
\end{equation*}
and
\begin{equation*}
\sup_{r\in[s',t]} \mc H(f^n_r)\leq \sup_{r\in[0,T]} \mc H(f_r)\leq C.
\end{equation*}

\smallskip
\noindent\emph{Step 2. Truncation of $\log$.}
The balance equation \eqref{beq} implies 
\begin{equation}\begin{split}\label{beqt0}
  &  \int\! \de x \int \pi(\de v) f_t^n(x,v) \phi(t,x,v)-\int\! \de x \int \pi(\de v) f_{s'}^n(x,v) \phi(s',x,v)
    \\ &-\int_{s'}^t \! \de r \int \! \de x \int \pi(\de v)
    f_r^n(x,v)
\big\{
    \partial_r \phi(r,x,v)
    %-\int_{s'}^t\! \de r \!\int\! \de x\! \int \pi(\de v) f_r^n(x,v)
    +b(v)\cdot \nabla \phi(r,x,v)
\big\}
    \\
 & =\frac 1 2 \int_{s'}^t\! \de r \int\! \de x\iint\! \pi(\de v)\pi (\de v')\eta^n(r,x,v,v')\big[\phi(r,x,v')- \phi(r,x,v) \big]
\end{split}\end{equation}  
for  all continuous functions
$\phi\colon [s',t]\times \bb T^d\times \mc V$ with compact support in $\mc V$ and continuously differentiable in  the first two variables. Recalling that $b$ has exponential moments, since $f^n$ has bounded entropy and $\eta^n\in L^1$ we can use $\phi$ bounded instead of compactly supported.

Given $0<\delta <L$ set
\begin{equation}\label{def:tlog}
  \log_{\delta,\, L}(u)= \begin{cases}
    \log \delta \quad \mbox{if } 0<u<\delta\\
    \log u \quad \mbox{if } \delta \leq  u \leq L\\
    \log L   \quad \mbox{if } u>L.
    \end{cases}
\end{equation}
By a straightforward approximation  we can choose as test function in \eqref{beqt0} $\phi=\log_{\delta,\, L}(f^n)$, obtaining
\begin{equation}\begin{split}\label{beqt}
    &  \int\! \de x \int \pi(\de v) f_t^n(x,v) \log_{\delta,\, L}(f^n_t(x,v))
    -\int\! \de x \int \pi(\de v) f_{s'}^n(x,v) \log_{\delta,\, L}(f^n_{s'}(x,v))
    \\ &-\int_{s'}^t \! \de r \int \! \de x \int \pi(\de v) \id_{[\delta,L]}(f^n_r(x,v))
\big\{
    \partial_r f^n_r(x,v)
    %-\int_{s'}^t\! \de r \!\int\! \de x\! \int \pi(\de v) \id_{[\delta,L]}(f^n_r(x,v))
    +b(v)\cdot \nabla f^n_r(x,v)
\big\}
    \\
    & =\frac 1 2 \int_{s'}^t\! \de r \int\! \de x\iint\! \pi(\de v)\pi (\de v')\eta^n(r,x,v,v')
    \big[\log_{\delta,\, L}(f^n_r(x,v))- \log_{\delta,\, L}(f^n_r(x,v')) \big].
\end{split}\end{equation}
We observe that
\begin{equation}\label{tra}\begin{split}
& \int_{s'}^t \! \de r \int \! \de x \int \pi(\de v) \id_{[\delta,L]}(f^n_r(x,v))
\big\{
    \partial_r f^n_r(x,v)
    %-\int_{s'}^t\! \de r \!\int\! \de x\! \int \pi(\de v) \id_{[\delta,L]}(f^n_r(x,v))
    +b(v)\cdot \nabla f^n_r(x,v)
    \big\}\\
    &   =\int \! \de x \int \pi(\de v) \big(f^n_t(x,v)\wedge \delta  \big)\vee L
    -\int \! \de x \int \pi(\de v) \big(f^n_{s'}(x,v)\wedge \delta  \big)\vee L.
    \end{split}\end{equation}  

\smallskip
\noindent\emph{Step 3. Removing the convolution.}
Since $\log_{\delta, L}$ is bounded, by dominated convergence we can remove regularization in space and time and we obtain
\begin{equation}\begin{split}\label{beqt1}
    &  \int\! \de x \int \pi(\de v) f_t(x,v) \log_{\delta,\, L}(f_t(x,v))
    -\int\! \de x \int \pi(\de v) f_{s'}(x,v) \log_{\delta,\, L}(f_{s'}(x,v))
    \\
 &- \int \! \de x \int \pi(\de v) \big(f_t(x,v)\wedge \delta  \big)\vee L
    +\int \! \de x \int \pi(\de v) \big(f_{s'}(x,v)\wedge \delta  \big)\vee L.
    \\
    & =\frac 1 2 \int_{s'}^t\! \de r \int\! \de x\iint\! \pi(\de v)\pi (\de v')\eta(r,x,v,v')
    \big[\log_{\delta,\, L}(f_r(x,v))- \log_{\delta,\, L}(f_r(x,v')) \big].
\end{split}\end{equation}

\smallskip
\noindent\emph{Step 3. Removing the truncation of $\log $.}

Here we take the limit $\delta\downarrow 0$ and $L\uparrow +\infty$ in \eqref{beqt1}.
For the left hand side this is accomplished by monotone convergence, in particular it converges
to $\mc H(P_t) - \mc H(P_{s'})$. For the right hand side, it is enough to show that
\begin{equation}\label{van}\begin{split}
    %\lim_{\delta, L }
    \frac 1 2 \int_{s'}^t\! \de r \int\! \de x\iint\! \pi(\de v)\pi (\de v')\eta(r,x,v,v')
    \Big\{ \big[\log_{\delta,\, L}(f_r(x,v))- \log(f_r(x,v)) \big]\\
    - \big[\log_{\delta,\, L}(f_r(x,v'))- \log(f_r(x,v')) \big] \Big\}
\end{split}\end{equation}
vanishes as $\delta \downarrow 0$ and $L\uparrow +\infty$.
We apply  Young inequality in the form
\begin{equation*}
p q \leq \psi_{\alpha}(p) +\psi_{\alpha}^*(q),
\end{equation*}  
where, for $\alpha\geq 0$,
\begin{equation*}
  \psi_\alpha(p)= p\varash \frac p \alpha - \sqrt{p^2 +\alpha^2} +\alpha,\quad
  \psi^*_\alpha(q)= \alpha \big(\cosh  q -1\big).
\end{equation*}
Observe that $\psi_\alpha$ and $\psi_\alpha^*$ are even.
By choosing $\alpha=2\sigma(v,v')\sqrt{f_r(x,v)f_r(x,v')}$, $p=\eta_r(x,v,v')$ and  $q=\frac 12 \big[\log_{\delta,\, L}(f_r(x,v))- \log(f_r(x,v)) 
  - \log_{\delta,\, L}(f_r(x,v'))+ \log(f_r(x,v')) \big]$, the first term is
\begin{equation*}\begin{split}
& \int_{s'}^t\! \de r \int\! \de x\iint\! \pi(\de v)\pi (\de v')
\psi_{\alpha}(\eta)\big[1-\id_{[\delta, L]}(f_r(x,v)) \big]\\
  \leq
&  \int_{s'}^t \!\de r \int\!\de x \iint \!\pi(\de v)\pi(\de v')\,
    \Psi_{\sigma(v,v')} \big( f_r(x,v),f_r(x,v');  \eta_r(x,v,v')
    \big),
\end{split}\end{equation*}  
which vanishes by dominated convergence since $\mathcal R (f,\eta)$ is finite. 
The second term has the following expression
\begin{equation*}\begin{split}
  & \int_{s'}^t\! \de r \int\! \de x\iint\! \pi(\de v)\pi (\de v') \psi_\alpha^*(q)\\
    &= \int_{s'}^t\! \de r \int\! \de x\iint\! \pi(\de v)\pi (\de v')
    \big[\id_{[0, \delta)}(f_r(x,v))\id_{[0, \delta)}(f_r(x,v'))\\
       &   \quad  + \id_{(L, +\infty)}(f_r(x,v))\id_{(L, +\infty)}(f_r(x,v'))   \big]
        \sigma(v, v')\big(\sqrt {f_r(x,v)}- \sqrt {f_r(x,v')}    \big)^2\\
        & +2\int_{s'}^t\! \de r \int\! \de x\iint\! \pi(\de v)\pi (\de v')
        \id_{[0, \delta)}(f_r(x,v))\id_{[\delta, L]}(f_r(x,v'))\\
          & \quad\times
          \sigma(v,v')\sqrt{f_r(x,v)f_r(x,v')}\Big( \frac {\sqrt \delta} {\sqrt{f_r(x,v)}}+\frac{\sqrt{f_r(x,v)}}{\sqrt \delta}-2  \Big)\\
        & +2\int_{s'}^t\! \de r \int\! \de x\iint\! \pi(\de v)\pi (\de v')
        \id_{[0, \delta)}(f_r(x,v))\id_{(L, +\infty)}(f_r(x,v'))\\
          & \quad\times
          \sigma(v,v')\sqrt{f_r(x,v)f_r(x,v')}
          \Big( \frac {\sqrt{L f_r(x,v)}}{\sqrt{\delta  f_r(x,v')}}
          +\frac {\sqrt{\delta f_r(x,v')}}{\sqrt{L  f_r(x,v)}}-2  \Big)\\
          & + 2 \int_{s'}^t\! \de r \int\! \de x\iint\! \pi(\de v)\pi (\de v')
        \id_{[\delta, L]}(f_r(x,v))\id_{(L, +\infty)}(f_r(x,v'))\\
        & \quad\times  \sigma(v,v')\sqrt{f_r(x,v)f_r(x,v')}
        \Big( \frac {\sqrt L} {\sqrt{f_r(x,v')}}+\frac{\sqrt{f_r(x,v')}}{\sqrt L}-2  \Big)
          \end{split}\end{equation*}    
We observe that  the first integral on the right hand side vanishes as $\delta\downarrow 0$, $L\uparrow +\infty$ since $\int_{s'}^t \de r\,  \mathcal E(f_r) < +\infty$. The term in the second integral
is bounded by $\sigma(v,v')\sqrt{\delta}\sqrt{f_r(x,v')}$ and the term in the third is bounded by $\sigma(v,v')\sqrt{\delta/L}f_r(x,v') $,  then the two integrals vanish as $\delta\downarrow 0$, $L\uparrow +\infty$ since the scattering rate $\lambda$ has all exponential moments and $f_r$ has finite entropy. Finally,  the term in the last integral in bounded by $\sigma(v,v')f_r(x,v')\id_{(L, +\infty)}(f_r(x,v'))$, which vanishes by dominated convergence.

\smallskip
Now we show that \eqref{jj1t} holds. For $s'>s\geq 0$ it  follows
 by applying again  Young inequality with $p=-\eta_r(x,v,v')$ and
$q=\frac 12 \big[\log (f_r(x,v))- \log(f_r(x,v'))\big]$ with the entropy balance \eqref{entbal}. Finally  the case $s'=s$ is achieved by the lower semi-continuity of $\mathcal H$.

\section{Bounds on $(-\mc L)^{-1}b$ for the Rayleigh gas}\label{s:A}

We identify the scattering kernel $\sigma$ on the right hand side of \eqref{ex3}.
Setting $z = v-v_1$, the collision operator in \eqref{ex3} becomes
\begin{equation}
\label{Lmio}
\mc Lf(v) =
\int_{\mathbb R^d}\de z\, h_\beta(v-z)\int_{S^{d-1}}\de\hat n \,
[\hat n\cdot z]_+\,\{f(v')-f(v)\}
\end{equation}
where 
$$v'= v - (\hat n \cdot z) \, \hat n$$
Fixed $\hat n$, we can write $z = \alpha \hat n + z^\perp$,
where $z^\perp$ 
lies in the hyperplane of dimension $d-1$ orthogonal to $\hat n$.
We indicate with $v^\perp =  v-(\hat n \cdot v)\, \hat n$,
the projection of $v$ on this hyperplane. 
We have
$[\hat n \cdot z]_+ = [\alpha]_+$, $\de z = \de\alpha \, \de z^\perp$, and
$$\mc Lf(v) = \int_{S^{d-1}} \de\hat n
\int \de z^\perp h_\beta^{d-1}(v^\perp-z^\perp)
 \int_0^{+\infty} \de \alpha \,\alpha \,h_\beta^1(v\cdot \hat n - \alpha) 
\{f(v')-f(v)\}
$$
where $h_\beta^k$ is the Maxwellian distribution in dimension $k$,
and  now $v' = v -\alpha \hat n$.
The integral in $\de z^\perp$ gives 1.
Choosing the new variable $w = v - \alpha \hat n$, we have
$\alpha = |v-w|$, $v\cdot \hat n= v\cdot (v-w)/|(v-w)|$, 
$$v\cdot \hat n - \alpha =  v\cdot (v-w)/|(v-w)| - |v-w| = 
w \cdot (v-w)/|v-w|$$
and  $\de w = \alpha^{d-1} \de\alpha\,\de\hat n = |v-w|^{d-1} \de\alpha \,
\de\hat n$.
Then
$$\mc Lf(v) =
\int_{\R^d}\de w 
 h_\beta^1( w\cdot (v-w) /|v-w|) \frac 1{|v-w|^{d-2}}\{f(w)-f(v)\}
$$
which is of the form \eqref{def:L} with 
\begin{equation}\label{sfuturo}\begin{aligned}\sigma(v,w) &=  \frac 1{|v-w|^{d-2}} 
\frac {h_\beta^1( w\cdot (v-w) /|v-w|)}{h_\beta(w)} \\
&= 
\left( \frac {\beta}{2\pi}\right)^{\frac {1-d}2}
\frac 1{|v-w|^{d-2}}  
\exp \left\{ 
\frac {\beta}{2} \frac {|w|^2|v-w|^2 - (w\cdot (v-w))^2}{|v-w|^2}\right\}\\
&= 
\left( \frac {\beta}{2\pi}\right)^{\frac {1-d}2}
\frac 1{|v-w|^{d-2}}  
\exp \left\{ 
\frac {\beta}{2} \frac {|w|^2|v|^2 - (w\cdot v)^2}{|v-w|^2}
\right\}.
\end{aligned}
\end{equation}
which is symmetric.
%
%|w|^2 |v-w|^2 - (w (v-w))^2
%|w|^2 |v-w|^2 - (wv - w^2)^2
%w^2 (v^2 + w^2 - 2vw) - (wv)^2 + 2wv w^2 - w^4
%w^2 v^2  - (wv)^2  
%|v|^2 |v-w| - (v (v-w))^2
We remark that for $d=3$ this expression has been obtained in \cite{LS}.

In order to prove that $\xi = -\mc L^{-1} v$ is bounded, we decompose
$\mc L$ in the gain and loss terms
$$-(\mc Lf)(v) = \lambda(v)  f(v) - (Gf)(v) $$
where 
$$\lambda(v)=(G1)(v) =
\int_{\mathbb R^d}\de w h_\beta(w)\sigma(v,w) =  
\chi\int_{\mathbb R^d}\de v_1 h_\beta(v_1)|v-v_1|
= \lambda(|v|),$$
and $\chi = \int_{S^{d-1}}\! \de\hat n\, [\hat n\cdot \hat v]_+$
for any unit vector $\hat v$. Observe that $(Gf)(v) = \lambda(v) (Kf)(v)$,
with $K$ defined in \eqref{def:K}.
Note that, for convexity, 
\begin{equation}
\label{dislambda}
\lambda(v) \ge \chi |v|.
\end{equation}
We search for a bounded function $\gamma(|v|)$ such that
$\xi (v) = \hat v \gamma(|v|)$.
%solves
%\begin{equation}
%\label{Lgv}
%v = -\mc L\xi (v) =  \hat v \lambda(v)  \gamma(|v|) - (G\xi)(v) 
% \end{equation}
Then we have
%From \eqref{def:L}
\begin{equation*}\begin{split} (G\xi)(v) &=
\int_{\R^d} \de w \, h_\beta(w) \sigma(v,w) 
{\gamma(|w|)} \frac w{|w|} \\
&=\int_{\R^d} \de w \, h_\beta(w) \sigma(v,w) 
\frac {\gamma(|w|)}{|w|}
\left[
(w - (\hat v \cdot w)\, 
\hat v)+ (\hat v \cdot  w )\, \hat v )\right]
\end{split}\end{equation*}
where in the last step we decomposed $w$ into the component along 
$\hat v$ and the orthogonal part
$w^\perp=w-(\hat v \cdot w) \hat v$.
Since $|w|$ and $\sigma(v,w)$ are invariant
in the exchange $w^\perp \to - w^\perp$,
then 
%$$(w\cdot v)^2,\ \ 
%|w|^2 = |w - (\hat v \cdot w)\, \hat v|^2 + |\hat v \cdot w|^2,\ \ 
%|v-w|^2 = |w - (\hat v \cdot w)\, \hat v|^2 + |\hat v \cdot (v-w)|^2
%$$
%Then
%&=
$$
(G\xi)(v) = \hat v \int_{\R^d}
\de w \, h_\beta(w) \sigma(v,w) (\hat v \cdot \hat w) 
\gamma(|w|). %=: \hat v \,(\tilde G\gamma) (|v|),
%\end{aligned}
$$
Since  the integral is invariant under rotations of $v$,
we can define the operator $\tilde G$ acting on functions on the positive
half line by
%\begin{equation*}
%  \begin{split}
$$  (\tilde Gf)(\rho) := 
\int_{\R^d} \de w \,
h_\beta(w) \sigma(\rho \hat v,w)  (\hat v \cdot \hat w)  f(|w|), $$
so that, for $\xi (v) = \hat v \gamma(|v|)$,
$(G\xi)(v) = \hat v (\tilde G\gamma) (|v|)$.
As
$$ (\tilde Gf)(\rho) = 
%
% \int_0^{+\infty} \de r \,g(|v|,r) \gamma(r),$$
\int_{w \cdot \hat v>0} 
\de w \, h_\beta(w) (\sigma(\rho \hat v,w)  -
\sigma(\rho \hat v,-w)) (\hat v \cdot w)
\gamma(|w|)$$%\end{split}\end{equation*}
and, if $w \cdot v>0$, then  $\sigma(v,w) \ge  \sigma(v,-w)$,
the operator $\tilde G$ has positive kernel.

%using that $|v-w|^2\le |v+w|^2$, we have
% $1/|v-w|^{d-2} \ge 1/|v+w|^{d-2}$ and
%$$(|w|^2|v|^2-(w\cdot v)^2) /|v-w|^2 \ge (|w|^2|v|^2-(w\cdot v)^2) /|v+w|^2.$$ 
%moreover
%$$\begin{aligned}
%&(w\cdot (v-w))^2/(v-w)^2 \le (w\cdot (v+w))^2/(v+w)^2\\
%&\iff (v+w)^2 (w\cdot (v-w)^2\le (v-w)^2(w\cdot (v+w))^2\\
%&\iff (w \cdot v) (w \cdot v)^2 + |w|^4) <  (v^2 + w^2) w^2 (w\cdot v) \\
%& \iff 0 <  (w \cdot v) (|w|^2(|v|^2 + |w|^2)  -  (w \cdot v)^2 - |w|^4) =
% (w \cdot v) |w|^2|v|^2 -  (w \cdot v)^2 \end{aligned}$$
%and the last inequality is satified if $w\cdot v > 0$.

Setting 
$\eta(\rho)=\lambda(\rho)\gamma(\rho)$, we look for the solution of the equation
\begin{equation}
\label{Arho=}
 \rho = \eta(\rho) - (A\eta) (\rho), \ \ \ \rho \in \R^+\end{equation}
%\imath = \eta - A\eta\end{equation}
where 
\begin{equation}
\label{AA}
\begin{aligned}
  (A\eta)(\rho) &= \left({\tilde G}\frac {\eta}{\lambda}\right)(\rho) = 
  \int_{\R^d} \de w \, h_\beta(w) \sigma(\rho \hat v,w)
  (\hat v \cdot \hat w)  
\frac {\eta(|w|)}{\lambda(|w|)} \\
&= 
\int_{\R^d}
\de v_1 \, h_\beta(v_1) \int_{S^{d-1}} \de \hat n \,
[\hat n \cdot (\rho \hat v-v_1)]_+ (\hat v \cdot \hat v') 
\frac {\eta(|v'|)}{\lambda(|v'|)}
\end{aligned}
\end{equation}
in which
$v'= \rho \hat v - (\hat n \cdot (\rho \hat v -v_1) ) \hat n$. 
The operator $A$ is self-adjoint 
with respect to the scalar product
$$(f,g) = \int_{\R^d}
\de v \, h_{\beta}(v) \frac 1{\lambda(|v|)} f(|v|)g(|v|),$$
defined for  $f,g:[0,\infty)\to \R$.
From the positivity of the kernel of the operator, 
it follows that if $f(\rho)\ge g(\rho)$ for any $\rho\ge 0$, 
then $(Af)(\rho) \ge (Ag) (\rho)$ for any $\rho\ge 0$. 
Moreover, if $\eta$ is continuous in $[0,\infty)$, $(A\eta)$ is continuous
in $[0,\infty)$
as follow from  \eqref{AA}.

By definition of $\lambda(\rho)$
$$(A\lambda) (\rho) = 
\int_{\R^d} \de w \, h_\beta(w) \sigma(\rho \hat v,w)
(\hat v \cdot \hat w)  
< \lambda(\rho)$$
and the inequality is strict for any $\rho$ 
because $\hat v \cdot \hat w < 1$
in a set of full measure.
Observe that, using the definition of $v'$,
$$ \hat v \cdot \hat v'=
\frac {\rho(1-(\hat n \cdot \hat v)^2) + (\hat n \cdot v_1) (\hat n \cdot
  \hat v)}{\sqrt{\rho^2 (1-(\hat n \cdot \hat v)^2) + (\hat n \cdot v_1)^2}}
$$
which for fixed $v_1$ converges to $\sqrt{1- (\hat n \cdot \hat v)^2}$
when $\rho\to +\infty$, while $[\hat n\cdot (\rho\hat v-v_1)]_+/\rho
\to |\hat n
\cdot \hat v]_+$.
By dominated convergence
$$\lim_{\rho\to +\infty} \frac 1{\rho} (A\lambda)(\rho) = 
\int_{\R^d}\de v_1 \, h_\beta(v_1)
\int_{S^{d-1}} \de \hat n \,
[\hat n \cdot \hat v]_+ \sqrt{1- (\hat n \cdot \hat v)^2} < \chi. 
$$
Since   
$\lim_{\rho \to +\infty} {\lambda(\rho)}/{\rho} = 
\int_{S^{d-1}}\de \hat n \,[\hat n \cdot \hat v]_+  = \chi$,
we 
then conclude that there exists a constant $0<z<1$ such that
$$(A\lambda) (\rho) < z \lambda(\rho)$$
for any $\rho \ge 0$.
%l-al > l-zl > (1-z) l > x(1-z) r
Since $\lambda(\rho) \ge \chi \rho$ (see \eqref{dislambda}),
if $\zeta = \chi (1-z)$, then 
$$\lambda(\rho) \ge \zeta \rho + (A \lambda) (\rho).$$
Denoting by $\mathop{\rm id}:\R^+ \to \R^+$ the identity function ${\rm id}(\rho) = \rho$,
and iterating  the above expression, we get
%$$\lambda (\rho) \ge 
%\zeta \rho + \zeta A \rho + A^2 \lambda (\rho) \ge  \dots \ge 
%\zeta\sum_{k=0}^n A^k \rho + (A^{n+1} \lambda)(\rho)$$
$$\lambda \ge 
\zeta \mathop{\rm id} + \zeta A \mathop{\rm id}
+ A^2 \lambda \ge  \dots \ge 
\zeta\sum_{k=0}^n A^k \mathop{\rm id}
+ A^{n+1} \lambda,$$
which implies that 
$\eta = \sum_{k=0}^{+\infty} A^k \mathop{\rm id}$
is a well defined, positive function, bounded by ${\lambda}/\zeta$.
Since $\eta$ solves \eqref{Arho=}, then 
$\xi(v) = \hat v{\eta(|v|)}/{\lambda(|v|)}$,
which is bounded by $1/\zeta$.

\section*{Acknowledgments}
We are grateful to Mauro Mariani for useful discussions about gradient flows and for his comments on an earlier version of the current manuscript.
L. Bertini acknowledges the support by the PRIN 20155PAWZB
``Large Scale Random Structures''.

\end{document}